\documentclass[reqno,12pt,letterpaper]{amsart}
\usepackage{amsmath,amssymb,amsthm,graphicx,mathrsfs,url}
\usepackage[usenames,dvipsnames]{color}
\usepackage[colorlinks=true,linkcolor=Red,citecolor=Green]{hyperref}

\long\def\switch#1#2{#1}

\def\smallsection#1{\smallskip\noindent\textbf{#1}.}

\newtheorem{theo}{Theorem}
\newtheorem{prop}{Proposition}[section]
\newtheorem{defi}[prop]{Definition}

\newtheorem{lemm}[prop]{Lemma}

\numberwithin{equation}{section}

\DeclareMathOperator{\Res}{Res}

\DeclareMathOperator{\comp}{comp}
\DeclareMathOperator{\const}{const}

\let\Im=\Imag

\let\Re=\Real
\DeclareMathOperator{\sgn}{sgn}

\DeclareMathOperator{\supp}{supp}
\DeclareMathOperator{\Vol}{Vol}
\DeclareMathOperator{\WF}{WF}

\def\WFh{\WF_h}

\title[Asymptotics of waves and resonances for black holes]%
{Asymptotics of linear waves and resonances\\
with applications to black holes}
\author{Semyon Dyatlov}
\email{dyatlov@math.berkeley.edu}
\address{Department of Mathematics, University of California,
Berkeley, CA 94720, USA}

\begin{document}

\begin{abstract}
We apply the results of~\cite{nhp} to describe asymptotic behavior of linear waves on stationary Lorentzian
metrics with $r$-normally hyperbolic trapped sets, in particular
Kerr and Kerr--de Sitter metrics with $|a|<M$ and $M\Lambda a\ll 1$.
We prove that if the initial data is localized
at frequencies $\sim\lambda\gg 1$, then the energy norm of the solution is bounded by
$\mathcal O(\lambda^{1/2} e^{-(\nu_{\min}-\varepsilon)t/2})+\mathcal O(\lambda^{-\infty})$,
for $t\leq C\log\lambda$, where $\nu_{\min}$ is a natural dynamical quantity.
The key tool is a microlocal projector splitting the solution into a component
with controlled rate of exponential decay and an $\mathcal O(\lambda e^{-(\nu_{\min}-\varepsilon)t})+\mathcal O(\lambda^{-\infty})$
remainder; this splitting can be viewed as an analog of resonance expansion.
Moreover, for the Kerr--de Sitter case we study quasi-normal modes; under a dynamical pinching condition,
a Weyl law in a band holds.
\end{abstract}

\maketitle

\addtocounter{section}{1}
\addcontentsline{toc}{section}{1. Introduction}

The subject of this \switch{paper}{chapter} are decay properties of solutions to the wave equation
for the rotating Kerr (cosmological constant $\Lambda=0$) and Kerr--de Sitter ($\Lambda>0$)
black holes, as well as for their stationary perturbations.
In the recent decade, there has been a lot of progress in understanding the upper
bounds on these solutions, producing a polynomial decay rate $\mathcal O(t^{-3})$
for Kerr and an exponential decay rate $\mathcal O(e^{-\nu t})$ for
Kerr--de Sitter (the latter is modulo constant functions). The weaker decay for $\Lambda=0$ is explained by the
presence of an asymptotically Euclidean infinite end; however, this polynomial decay
comes from low frequency contributions.

We instead concentrate on the decay of solutions
with initial data localized at \emph{high frequencies} $\sim \lambda\gg 1$; it is related
to the geometry of the \emph{trapped set} $\widetilde K$, consisting of lightlike geodesics
that never escape to the spatial infinity or through the event horizons. The trapped
set for both Kerr and Kerr--de Sitter metrics is \emph{$r$-normally hyperbolic},
and this dynamical property is stable under stationary perturbations
of the metric~-- see~\S\ref{s:kdsu-stability}.
The key quantities associated to such trapping are the minimal and maximal transversal expansion rates
$0< \nu_{\min}\leq \nu_{\max}$, see~\eqref{e:tilde-nu-min}, \eqref{e:tilde-nu-max}.
Using \switch{our recent work~\cite{nhp},}{Chapter~\ref{c:nhp},} we show the exponential decay rate
$\mathcal O(\lambda^{1/2} e^{-(\nu_{\min}-\varepsilon)t/2})+\mathcal O(\lambda^{-\infty})$, valid for $t=\mathcal O(\log\lambda)$
(Theorem~\ref{t:decay}).
This bound is new for the Kerr case, complementing Price's law.

Our methods give a more precise microlocal description
of long time propagation of high frequency solutions.
In Theorem~\ref{t:res-dec}, we split a solution $u(t)$ into two approximate solutions to the
wave equation, $u_\Pi(t)$ and $u_R(t)$, with the rate of decay of $u_\Pi(t)$ between
$e^{-(\nu_{\max}+\varepsilon)t/2}$ and $e^{-(\nu_{\min}-\varepsilon)t/2}$
and $u_R(t)$ bounded from above by $\lambda e^{-(\nu_{\min}-\varepsilon)t}$, all modulo
$\mathcal O(\lambda^{-\infty})$ errors. This splitting is achieved using the Fourier integral
operator $\Pi$ constructed in\switch{~\cite{nhp}}{ Chapter~\ref{c:nhp}} using the global dynamics of the flow.

For the $\Lambda>0$ case, we furthermore study resonances, or \emph{quasi-normal modes},
the complex frequencies $z$ of solutions to the wave equation of the form $e^{-izt}v(x)$. Under a pinching
condition $\nu_{\max}<2\nu_{\min}$ which is numerically verified to be true for a large range of parameters
(see Figure~\ref{f:kdsu-nums}(a)), we show existence of a band of quasi-normal modes
satisfying a Weyl law~-- Theorem~\ref{t:kds-weyl}. In particular, this provides a large family of
exact high frequency solutions to the wave equation that decay no faster than $e^{-(\nu_{\max}+\varepsilon)t/2}$.
We finally compare
our theoretical prediction on the imaginary parts of resonances in the band with
the exact quasi-normal modes for Kerr computed by the authors of~\cite{bcs}, obtaining
remarkable agreement~-- see Figure~\ref{f:kdsu-nums}(b).

Theorems~\ref{t:decay}--\ref{t:kds-weyl}
are related to the resonance expansion and quantization condition proved for the slowly rotating Kerr--de Sitter
in\switch{~\cite{zeeman}.}{ Chapter~\ref{c:zeeman}.}
In this \switch{paper}{chapter}, however, we use dynamical assumptions stable
under perturbations, rather than complete integrability of geodesic flow on Kerr(--de Sitter),
and do not recover the precise results of\switch{~\cite{zeeman}.}{ Chapter~\ref{c:zeeman}.}

\smallsection{Statement of results}
The Kerr(--de Sitter) metric, described in detail in~\S\ref{s:kdsu-the-metric}, depends
on three parameters, $M$ (mass), $a$ (speed of rotation), and $\Lambda$ (cosmological constant).
We assume that the dimensionless quantities $a/M$ and $\Lambda M^2$ lie in a small
neighborhood (see Figure~\ref{f:kdsu-gaps}(a)) of either the Schwarzschild(--de Sitter) case,
\begin{equation}
  \label{e:sds-bound}
a=0,\quad
9\Lambda M^2<1,
\end{equation}
or the subextremal Kerr case
\begin{equation}
  \label{e:kerr-bound}
\Lambda=0,\quad
|a|<M.
\end{equation}
Our results apply as long as certain dynamical conditions are satisfied, and
likely hold for a larger range of parameters, see the remark following Proposition~\ref{l:kds-trapping}.
To facilitate the discussion of perturbations, we adopt the abstract framework
of~\S\ref{s:kdsu-assumptions}, with the spacetime $\widetilde X_0=\mathbb R_t\times X_0$
and a Lorentzian metric $\tilde g$
on $\widetilde X_0$ which is stationary in the sense that $\partial_t$
is Killing. The space slice $X_0$ is noncompact because
of the spatial infinity and/or event horizon(s); to measure
the distance to those, we use
a function $\mu\in C^\infty(X_0;(0,\infty))$, such that
$X_{\delta}:=\{\mu>\delta\}$ is compact for each $\delta>0$. For the exact
Kerr(--de Sitter metric), the function $\mu$ is defined in~\eqref{e:mu-defined}.

We study solutions to the wave equation in $\widetilde X_0$,
\begin{equation}
  \label{e:kds-we}
\Box_{\tilde g}u(t)=0,\ t\geq 0;\quad
u|_{t=0}=f_0,\
\partial_t u|_{t=0}=f_1,
\end{equation}
with $f_0,f_1\in C_0^\infty(X_0)$ and
the time variable shifted so that the metric continues smoothly past
the event horizons~-- see~\eqref{e:wewewe}.
To simplify the statements, and because our work focuses on the phenomena driven by
the trapped set, we only study the behavior of solutions
in $X_{\delta_1}$ for some small $\delta_1>0$. 
Define the energy norm%
\begin{equation}
  \label{e:kdsu-energy}
\|u(t)\|_{\mathcal E}:=\|u(t)\|_{H^1(X_{\delta_1})}+\|\partial_t u(t)\|_{L^2(X_{\delta_1})}.
\end{equation}
\begin{theo}
  \label{t:decay}
Fix $T,N>0$, $\varepsilon,\delta_1>0$, and 
let $(\widetilde X_0,\tilde g)$ be the Kerr(--de Sitter) metric
with $M,a,\Lambda$ near one of the cases \eqref{e:sds-bound}
or~\eqref{e:kerr-bound}, or its small stationary perturbation as discussed
in~\S\ref{s:kdsu-stability} (the maximal size of the perturbation depending on $T,N$).

Assume that $f_0(\lambda),f_1(\lambda)\in C_0^\infty(X_{\delta_1})$ are localized at frequency~$\sim\lambda\to\infty$
in the sense of~\eqref{e:kdsu-localization}.
Then the solution~$u_\lambda$ to~\eqref{e:kds-we} with initial data $f_0,f_1$
satisfies the bound
\begin{equation}
  \label{e:decay}
\|u_\lambda (t)\|_{\mathcal E}\leq C (\lambda^{1/2}e^{-(\nu_{\min}-\varepsilon)t/2}+\lambda^{-N})\|u_\lambda(0)\|_{\mathcal E},\quad
0\leq t\leq T\log\lambda.
\end{equation}
\end{theo}
Here we say that $f=f(\lambda)$ is localized at frequencies $\sim\lambda$,
if for each coordinate neighborhood $U$ in $X_0$ and each
$\chi\in C_0^\infty(U)$, the Fourier transforms $\widehat{\chi f}(\xi)$ in the corresponding
coordinate system satisfy for each $N$,
\begin{equation}
  \label{e:kdsu-localization}
\int_{\mathbb R^3\setminus\{C_{U,\chi}^{-1}\leq |\xi|\leq C_{U,\chi}\}}\langle\xi\rangle^N|\widehat{\chi f}(\xi)|^2\,d\xi
=\mathcal O(\lambda^{-N}),
\end{equation}
where $C_{U,\chi}>0$ is a constant independent of $\lambda$.
For the proof, it is more convenient to use \emph{semiclassical rescaling} of frequencies
$\xi\mapsto h\xi$, where $h=\lambda^{-1}\to 0$ is the semiclassical parameter,
and the notion of $h$-wavefront set $\WFh(f)\subset T^*X_0$.
The requirement that $f_j$ is microlocalized at frequencies $\sim h^{-1}$
is then equivalent to stating that $\WFh(f_j)$ is contained in a fixed compact subset of $T^*X_0\setminus 0$,
with $0$ denoting the zero section; see~\S\ref{s:kdsu-prelims} for details.

The main component of the proof of Theorem~\ref{t:decay} is the following
\begin{theo}
  \label{t:res-dec}
Under the assumptions of Theorem~\ref{t:decay}, for each families $f_0(h),f_1(h)\in C_0^\infty(X_{\delta_1})$
with $\WFh(f_j)$ contained in a fixed compact subset of $T^*X_0\setminus 0$ and $u(h)$ the corresponding solution to~\eqref{e:kds-we},
for $t_0$ large enough there exists a decomposition
$$
u(t,x)=u_\Pi(t,x)+u_R(t,x),\quad
t_0\leq t\leq T\log(1/h),\
x\in X_{\delta/2},
$$
such that $\Box_{\tilde g}u_\Pi(t),\Box_{\tilde g}u_R(t)$
are $\mathcal O(h^N)_{H^N_h}$ on $X_{\delta_1}$ uniformly in $t\in [t_0,T\log(1/h)]$,
and we have uniformly in $t_0\leq t\leq T\log(1/h)$,
\begin{align}
  \label{e:decay-1}
\|u_\Pi(t_0)\|_{\mathcal E}&\leq Ch^{-1/2}\|u(0)\|_{\mathcal E},\\
  \label{e:decay-2}
\|u_\Pi(t)\|_{\mathcal E}&\leq Ce^{-(\nu_{\min}-\varepsilon)t/2}\|u_\Pi(t_0)\|_{\mathcal E}+Ch^N\|u(0)\|_{\mathcal E},\\
  \label{e:decay-3}
\|u_\Pi(t)\|_{\mathcal E}&\geq C^{-1}e^{-(\nu_{\max}+\varepsilon)t/2}\|u_\Pi(t_0)\|_{\mathcal E}-Ch^N\|u(0)\|_{\mathcal E},\\
  \label{e:decay-4}
\|u_R(t)\|_{\mathcal E}&\leq C(h^{-1}e^{-(\nu_{\min}-\varepsilon)t}+h^N)\|u(0)\|_{\mathcal E}.
\end{align}
\end{theo}
The decomposition $u=u_\Pi+u_R$ is achieved in~\S\ref{s:kdsu-decay} using the Fourier integral operator $\Pi$
constructed for $r$-normally hyperbolic trapped sets in\switch{~\cite{nhp}}{ Chapter~\ref{c:nhp}}. The component $u_\Pi$
enjoys additional microlocal properties, such as localization on the outgoing tail
and approximately solving a pseudodifferential equation~-- see the proof
of Theorem~\ref{t:internal-waves} in~\S\ref{s:kdsu-decay}
and~\switch{\cite[\S8.5]{nhp}}{\S\ref{s:resonant-states}}.
We note that~\eqref{e:decay-3} gives a \emph{lower bound} on the rate of decay of the approximate
solution $u_\Pi$,
if $\|u_\Pi(t_0)\|_{\mathcal E}$ is not
too small compared to $\|u(0)\|_{\mathcal E}$, and the existence of a large family
of solutions with the latter property follows from the construction of $u_\Pi$.
\begin{figure}
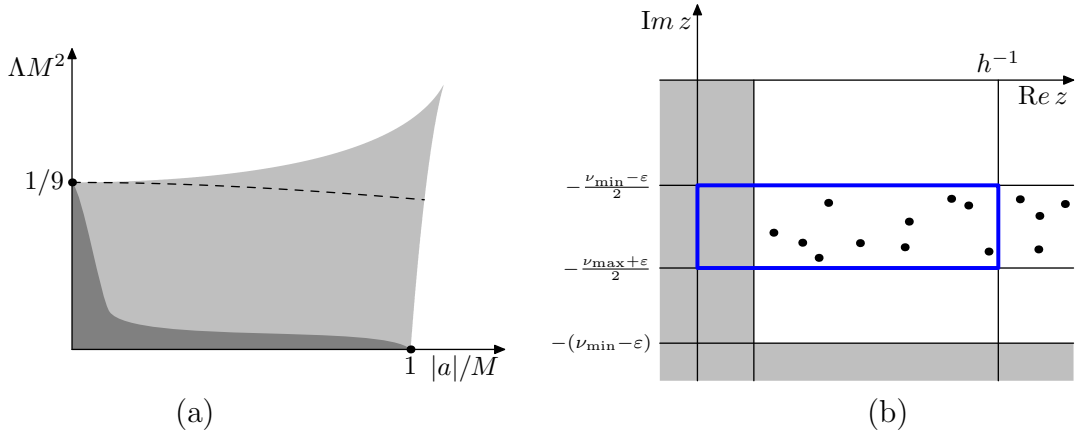

\includegraphics{kdsu.12}
\quad
\includegraphics{kdsu.8}
\hbox to\hsize{\hss (a) \hss\hss\hss (b)\hss}
\caption[(a) The numerically computed admissible range of parameters for the subextremal
Kerr--de Sitter black hole (b) An illustration of Theorem~\ref{t:kds-weyl}.]%
{(a) The numerically computed admissible range of parameters for the subextremal
Kerr--de Sitter black hole (light shaded) and a schematic depiction of the range of parameters to which
our results apply (dark shaded). The region above the dashed line is where the
resolvent does not admit a meromorphic continuation, see~\eqref{e:weird-region}.
(b) An illustration of Theorem~\ref{t:kds-weyl}; \eqref{e:kdsu-weyl} counts resonances in the outlined box
and the unshaded regions above and below the box represent~\eqref{e:kdsu-gap}.}
\label{f:kdsu-gaps}
\end{figure}
We remark that Theorems~\ref{t:decay} and~\ref{t:res-dec} are completely independent
from the behavior of linear waves at low frequency. In fact, we do not even
use the boundedness in time of solutions for the wave equation, assuming merely
that they grow at most exponentially (which is trivially true in our case); this suffices since $\mathcal O(h^\infty)$
remainders overcome such growth for $t=\mathcal O(\log(1/h))$. If a boundedness
statement is available, then our results can be extended to all times, though
the corresponding rate of decay stays fixed for $t\gg\log(1/h)$ because
of the $\mathcal O(h^\infty)$ term.

To formulate the next result, we restrict to the case $\Lambda>0$,
or its small stationary perturbation. In this case, the metric has two event horizons
and we consider the discrete set $\Res$ of \emph{resonances},
as defined for example in~\cite{vasy1}.
As a direct application of\switch{~\cite[Theorems~1 and~2]{nhp},}{ Theorems~\ref{t:gaps} and~\ref{t:weyl-law},}
we obtain two gaps and a band of resonances in between with a Weyl law (see Figure~\ref{f:kdsu-gaps}(b)):
\begin{theo}
  \label{t:kds-weyl}
Let $(\widetilde X_0,\tilde g)$ be the Kerr--de Sitter metric
with $M,a,\Lambda$ near one of the cases \eqref{e:sds-bound}
or~\eqref{e:kerr-bound} and $\Lambda>0$, or its small stationary perturbation as discussed
in~\S\ref{s:kdsu-stability}. Fix $\varepsilon>0$. Then:

1. For $h$ small enough, there are no resonances in the region
\begin{equation}
  \label{e:kdsu-gap}
\{|\Re z|\geq h^{-1},\
\Im z\in [-(\nu_{\min}-\varepsilon), 0]\setminus {\textstyle{1\over 2}}(-(\nu_{\max}+\varepsilon),-(\nu_{\min}-\varepsilon))\}
\end{equation}
and the corresponding semiclassical scattering resolvent,
namely the inverse of the operator~\eqref{e:semicr}, is bounded by $Ch^{-2}$
for $z$ in this region.

2. Under the pinching condition
\begin{equation}
  \label{e:kds-pinching}
\nu_{\max}<2\nu_{\min}
\end{equation}
and for $\varepsilon$ small enough so that $\nu_{\max}+\varepsilon<2(\nu_{\min}-\varepsilon)$,
we have the Weyl law
\begin{equation}
  \label{e:kdsu-weyl}
\#(\Res\cap \{0\leq\Re z\leq h^{-1},\ \Im z\in [-(\nu_{\min}-\varepsilon),0]\})=\\
(2\pi h)^{1-n}(c_{\widetilde K}+o(1))
\end{equation}
as $h\to 0$, where $c_{\widetilde K}$ is the symplectic volume of a certain part of the trapped set $\widetilde K$,
see~\eqref{e:c-k-tilde} and~\eqref{e:c-tilde-k-2}.
\end{theo}
The pinching condition~\eqref{e:kds-pinching} is true for the non-rotating case $a=0$,
since $\nu_{\min}=\nu_{\max}$ there (see Proposition~\ref{l:special-sds}).
However, it is violated for the nearly extremal case $M-|a|\ll M$, at least for $\Lambda$
small enough; in fact, as $|a|/M\to 1$, $\nu_{\max}$ stays bounded away from zero,
while $\nu_{\min}$ converges to zero~-- see Proposition~\ref{l:special-kerr} and Figure~\ref{f:kdsu-nums}(a).
\begin{figure}
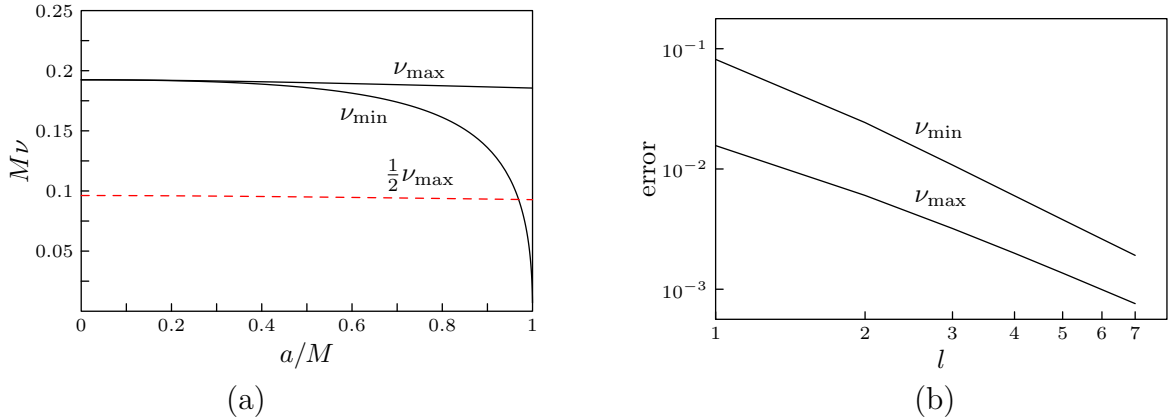

\includegraphics{kdsu.1}
\qquad\quad
\includegraphics{kdsu.7}
\hbox to\hsize{\hss (a)\hss\hss\hss (b)\hss}
\caption[(a) Numerically computed minimal and maximal transversal expansion rates
(b) A log-log plot of the relative errors of minimal/maximal imaginary parts of resonances.]%
{(a) Numerically computed minimal and maximal transversal expansion rates for $\Lambda=0$
and the range of $a$ for which~\eqref{e:kds-pinching} holds.
(b) A log-log plot of the relative errors ${|\nu_{\min}^R(l)-\nu_{\min}/2|\over \nu_{\min}^R(l)},
{|\nu_{\max}^R(l)-\nu_{\max}/2|\over\nu_{\max}^R(l)}$, where $\nu_{\min}^R(l),\nu_{\max}^R(l)$ are
the minimal/maximal imaginary
parts of resonances in the first band defined by~\eqref{e:nu-R}.}
\label{f:kdsu-nums}
\end{figure}
Note that $(\nu_{\min}-\varepsilon)/2$ is the size of the resonance free strip
and thus gives the minimal rate of exponential decay of linear waves on Kerr--de Sitter,
modulo terms coming from finitely many resonances, by means of a resonance expansion~-- see for
example~\cite[Lemma~3.1]{vasy1}.

To demonstrate the sharpness of the size of the band of resonances
$\{\Im\omega\in {1\over 2}[-\nu_{\max}-\varepsilon,-\nu_{\min}+\varepsilon]\}$,
we use the exact quasi-normal modes for the Kerr metric computed (formally, since
one cannot meromorphically continue the resolvent in the $\Lambda=0$ case; however,
one could consider the case of a very small positive $\Lambda$) by
Berti--Cardoso--Starinets~\cite{bcs}. Similarly to the quantization
condition of\switch{~\cite{zeeman},}{ Chapter~\ref{c:zeeman},} these resonances $\omega_{mlk}$ are indexed by
three integer parameters $m\geq 0$ (depth), $l\geq 0$ (angular energy),
and $k\in [-l,l]$ (angular momentum). The parameter $l$ roughly corresponds
to the real part of the resonance and the parameter $m$, to its imaginary part.
We define
\begin{equation}
  \label{e:nu-R}
\nu_{\min}^R(l):=\min_{k\in [-l,l]} (-\Im \omega_{0lk}),\quad
\nu_{\max}^R(l):=\max_{k\in [-l,l]} (-\Im \omega_{0lk}).
\end{equation}
We compare $\nu_{\min}^R(l),\nu_{\max}^R(l)$ with $\nu_{\min}/2,\nu_{\max}/2$
and plot the supremum of the relative error over $a/M\in [0,0.95]$
for different values of $l$; the error decays like $\mathcal O(l^{-1})$~-- see Figure~\ref{f:kdsu-nums}(b).

\smallsection{Previous work}
We give an overview of results on decay and non-decay on black hole backgrounds;
for a more detailed discussion of previous results on normally hyperbolic trapped sets and
resonance asymptotics, see \switch{the introduction to~\cite{nhp}.}{\S\ref{s:nhp-intro}.}

The study of boundedness of solutions to the wave equation for the Schwarzschild ($\Lambda=a=0$)
black hole was initiated in~\cite{wald,kay-wald} and decay results for this case
have been proved in~\cite{bl-st,da-ro-redshift,mmtt,luk1}. The slowly rotating Kerr case ($\Lambda=0,|a|\ll M$)
was considered in~\cite{an-bl,da-ro-boundedness,da-ro-lectures,ta-price,ta-to,tohaneanu,me-ta-to,luk2},
and the full subextremal Kerr case ($\Lambda=0,|a|<M$) in~\cite{fksy,fksy:erratum,da-ro-smalla,da-ro-proceedings,ysr1,ysr2}~--
see~\cite{da-ro-proceedings} for a more detailed overview. In either case the decay is polynomial in time,
with the optimal decay rate $\mathcal O(t^{-3})$. A decay rate of $\mathcal O(t^{-2l-3})$, known as~\emph{Price's Law},
was proved in~\cite{dss,dss1} for linear waves on the Schwarzschild black hole for solutions living
on the $l$-th spherical harmonic; the constant in the $\mathcal O(\cdot)$ depends on $l$.
Our Theorem~\ref{t:decay} improves on these decay rates in the high frequency regime $l=\lambda\gg 1$, for times $\mathcal O(\log\lambda)$.

The extremal Kerr case ($\Lambda=0,|a|=M$) was recently studied for axisymmetric solutions
in~\cite{aretakis3}, with a weaker upper bound due to the degeneracy of the event horizon.
The earlier work~\cite{aretakis1,aretakis2} suggests that one cannot expect the $\mathcal O(t^{-3})$ decay
to hold in the extremal case. In the high frequency regime studied here,
we do not expect to get exponential decay due to the presence of slowly damped geodesics near
the event horizon, see Figure~\ref{f:kdsu-nums}(a) above.

The Schwarzschild--de Sitter case ($\Lambda>0,a=0$) was considered in~\cite{sb-zw,bo-ha,da-ro-sds,me-sb-va},
proving an exponential decay rate at all frequencies, a quantization condition for resonances, and a resonance expansion,
all relying on separation of variables techniques. \switch{In~\cite{skds,xpd,zeeman},}{In Chapters~\ref{c:skds}
and~\ref{c:zeeman} and~\cite{xpd},} a same flavor of results
was proved for the slowly rotating Kerr--de Sitter ($\Lambda>0,|a|\ll M$). The problem was then studied
from a more geometric perspective, aiming for results that do not depend on symmetries and apply to
perturbations of the metric~-- the resonance free strip of~\cite{wu-zw} for normally hyperbolic trapping,
the gluing method of~\cite{da-va1}, and the analysis of the event horizons and low frequencies of~\cite{vasy1}
together give an exponential decay rate which is stable under perturbations, for $\Lambda>0,|a|<{\sqrt{3}\over 2}M$,
provided that there are no resonances in the upper half-plane except for the resonance at zero.
Our Theorem~\ref{t:kds-weyl} provides detailed information on the behavior of resonances below
the resonance free strip of~\cite{wu-zw}, without relying on the symmetries of the problem.

Finally, we mention the the Kerr--AdS case ($\Lambda<0$). The metric in this case exhibits strong
(elliptic) trapping, which suggests that the decay of linear waves is very slow because of the high frequency
contributions. A logarithmic upper bound was proved in~\cite{ho-sm}, and existence of resonances exponentially
close to the real axis and a logarithmic lower bound were established in~\cite{oran,ho-sm-new}.

Quasi-normal modes (QNMs) of black holes have a rich history of study in the physics literature,
see~\cite{ko-sc}. The exact QNMs of Kerr black holes were computed in~\cite{bcs},
which we use for Figure~\ref{f:kdsu-nums}(b). The high-frequency approximation for QNMs,
using separation of variables and WKB techniques, has been obtained in~\cite{ynzzzc,yzznbc,hod}.
In particular, for the nearly extremal Kerr case their size of the resonance free
strip agrees with Proposition~\ref{l:special-kerr}; moreover, they find
a large number of QNMs with small imaginary parts, which correspond to a positive
proportion of the Liouville tori on the trapped set lying close to the event horizon.
See~\cite{ynzzzc} for an overview of the recent physics literature on the topic.
We finally remark that the speed of rotation of an astrophysical black hole
(NGC 1365) has recently been accurately measured in~\cite{nature}, yielding
a high speed of rotation: $a/M\geq 0.84$ at 90\% confidence.

\smallsection{Structure of the \switch{paper}{chapter}}
In~\S\ref{s:kdsu-general}, we study semiclassical properties
of solutions to the wave equation on stationary Lorentzian metrics
with noncompact space slices. We operate under the geometric and dynamical assumptions
of~\S\ref{s:kdsu-assumptions}; while
these assumptions are motivated by Kerr(--de Sitter) metrics and their stationary
perturbations, no explicit mention of these metrics is made.
The analysis of~\S\ref{s:kdsu-general} works in a fixed compact subset of the space
slice, and the results apply under microlocal assumptions in this compact subset (namely, outgoing
property of solutions to the wave equation for Theorems~\ref{t:decay}--\ref{t:res-dec}
and meromorphic continuation of the scattering resolvent with an outgoing parametrix
for Theorem~\ref{t:kds-weyl})
which are verified for our specific applications in~\S\ref{s:kdsu-waves} and~\S\ref{s:kdsu-resonances}. In~\S\ref{s:kdsu-space},
we reduce the problem to the space slice via the stationary d'Alembert--Beltrami operator
and show that some of the assumptions of~\switch{\cite[\S\S4.1, 5.1]{nhp}}{\S\S\ref{s:framework-assumptions}, \ref{s:dynamics}} are satisfied.
In~\S\ref{s:kdsu-decay}, we use the methods of\switch{~\cite{nhp}}{ Chapter~\ref{c:nhp}} to prove asymptotics of outgoing
solutions to the wave equation.

Next, \S\ref{s:kds-metrix} contains the applications of\switch{~\cite{nhp}}{ Chapter~\ref{c:nhp}} and~\S\ref{s:kdsu-general}
to the Kerr(--de Sitter) metrics and their perturbations. In~\S\ref{s:kdsu-the-metric},
we define the metrics and establish their basic properties, verifying
in particular the geometric assumptions of~\S\ref{s:kdsu-assumptions}.
In~\S\ref{s:kdsu-trapping}, we show that the trapping is $r$-normally hyperbolic,
verifying the dynamical assumptions of~\S\ref{s:kdsu-assumptions}.
In~\S\ref{s:kdsu-trapping-special}, we study in greater detail trapping in the Schwarzschild(--de Sitter)
case $a=0$ and in the nearly extremal Kerr case $\Lambda=0,a=M-\epsilon$, in particular
showing that the pinching condition~\eqref{e:kds-pinching} is violated for the latter case;
we also study numerically some properties of the trapping for the general Kerr case,
and give a formula for the constant in the Weyl law.
In~\S\ref{s:kdsu-waves}, we study solutions to the wave equation on Kerr(--de Sitter),
using the results of~\S\ref{s:kdsu-decay} to prove Theorems~\ref{t:decay} and~\ref{t:res-dec}.
In~\S\ref{s:kdsu-resonances}, we use the results of\switch{~\cite{nhp}}{ Chapter~\ref{c:nhp}} and~\cite{vasy1}
to prove Theorem~\ref{t:kds-weyl} for Kerr--de Sitter. Finally,
in~\S\ref{s:kdsu-stability}, we explain why our results apply to small smooth
stationary perturbations of Kerr(--de Sitter) metrics.

\section{General framework for linear waves}
  \label{s:kdsu-general}

\subsection{Semiclassical preliminaries}
  \label{s:kdsu-prelims}

We start by briefly reviewing some notions of semiclassical analysis,
following~\switch{\cite[\S3]{nhp}}{\S\ref{s:prelim}}. For a detailed introduction to the subject, the reader
is directed to~\cite{e-z}.

Let $X$ be an $n$-dimensional manifold without boundary. Following~\switch{\cite[\S3.1]{nhp},}{\S\ref{s:prelim-basics},}
we consider the class $\Psi^k(X)$ of all semiclassical pseudodifferential
operators with classical symbols of order $k$. If $X$ is noncompact,
we impose no restrictions on how fast the corresponding symbols can grow
at spatial infinity. The microsupport of a pseudodifferential operator $A\in\Psi^k(X)$,
also known as its $h$-wavefront set $\WFh(A)$, is a closed subset of the
fiber-radially compactified cotangent bundle $\overline T^*X$. We denote
by $\Psi^{\comp}(X)$ the class of all pseudodifferential operators
whose wavefront set is a compact subset of $T^*X$ (and in particular
lies away from the fiber infinity). Finally, we say that $A=\mathcal O(h^\infty)$
microlocally in some open set $U\subset \overline T^*X$, if $\WFh(A)\cap U=\emptyset$;
similar notions apply to tempered distributions and operators below.

Using pseudodifferential operators, we can study microlocalization of $h$-tempered
distributions, namely families of distributions $u(h)\in\mathcal D'(X)$
having a polynomial in $h$ bound in some Sobolev norms on compact sets,
by means of the wavefront set $\WFh(u)\subset \overline T^*X$.
Using Schwartz kernels, we can furthermore study $h$-tempered operators
$B(h):C_0^\infty(X_1)\to \mathcal D'(X_2)$ and their wavefront sets
$\WFh(B)\subset\overline T^*(X_1\times X_2)$. Besides pseudodifferential
operators (whose wavefront set is this framework is the image under the diagonal embedding
$\overline T^*X\to\overline T^*(X\times X)$ of the wavefront set used in the previous paragraph)
we will use the class $I^{\comp}(\Lambda)$ of compactly supported
and compactly microlocalized Fourier integral operators associated to some
canonical relation $\Lambda\subset T^*(X\times X)$, see~\switch{\cite[\S3.2]{nhp};}{\S\ref{s:prelim-fio};} for
$B\in I^{\comp}(\Lambda)$, $\WFh(B)\subset\Lambda$ is compact.

The $h$-wavefront set of an $h$-tempered family of distributions $u(h)$
can be characterized using the semiclassical Fourier transform
$$
\mathcal F_h v(\xi)=(2\pi h)^{-n/2}\int_{\mathbb R^n} e^{-{i\over h}x\cdot\xi} v(x)\,dx,\
v\in \mathscr S'(\mathbb R^n).
$$
We have $(x,\xi)\not\in\WFh(u)$ if and only if there exists a coordinate neighborhood
$U_x$ of $x$ in $X$, a function $\chi\in C_0^\infty(U_x)$
with $\chi(x)\neq 0$, and a neighborhood $U_\xi$ of $\xi$ in $\overline T_x^*X$
such that if we consider $\chi u$ as a function on $\mathbb R^n$ using the corresponding
coordinate system, then for each $N$,
\begin{equation}
  \label{e:wf-fourier}
\int_{U_\xi} \langle\xi\rangle^N|\mathcal F_h(\chi u)(\xi)|^2\,d\xi=\mathcal O(h^N).
\end{equation}
The proof is done analogously to~\cite[Theorem~18.1.27]{ho3}.

One additional concept that we need is microlocalization of distributions
depending on the time variable that varies in a set whose size can grow with $h$.
Assume that $u(t;h)$ is a family of distributions on $(-\varepsilon,T(h)+\varepsilon)\times X$,
where $\varepsilon>0$ is fixed and $T(h)>0$ depends on $h$. For $s\in [0,T(h)]$, define
the shifted function
$$
u_s(t;h)=u(s+t;h),\quad
t\in (-\varepsilon,\varepsilon),
$$
so that $u_s\in \mathcal D'((-\varepsilon,\varepsilon)\times X)$ is a distribution
on a time interval independent of $h$. We then say that
$u$ is $h$-tempered uniformly in $t$, if $u_s$ is $h$-tempered uniformly in $s$,
that is, for each $\chi\in C_0^\infty((-\varepsilon,\varepsilon)\times X)$, there exist
constants $C$ and $N$ such that $\|\chi u_s\|_{H_h^{-N}}\leq Ch^{-N}$ for all
$s\in [0,T(h)]$. Next, we define the projected wavefront set $\widetilde\WFh(u)\subset \overline{T^*X\times \mathbb R_\tau}$,
where $\tau$ is the momentum corresponding to $t$ and $\overline{T^*X\times\mathbb R_\tau}$
is the fiber-radial compactification of the vector bundle $T^*X\times\mathbb R_\tau$, with
$\mathbb R_\tau$ part of the fiber, as follows: $(x,\xi,\tau)$ does not lie in $\widetilde\WFh(u)$
if and only if there exists a neighborhood $U$ of $(x,\xi,\tau)$ in $\overline{T^*X\times \mathbb R_\tau}$
such that
$$
\sup_{s\in [0,T(h)]}\|Au_s\|_{L^2}=\mathcal O(h^\infty)
$$
for each compactly supported $A\in\Psi^{\comp}((-\varepsilon,\varepsilon)\times X)$
such that $\WFh(A)\cap ((-\varepsilon,\varepsilon)_t\times U)=\emptyset$. If $T(h)$ is independent
of $h$, then $\widetilde\WFh(u)$ is simply the closure of the projection of $\WFh(u)$ onto
the $(x,\xi,\tau)$ variables. The notion of $\widetilde\WFh$ makes it possible to talk
about $u$ being microlocalized inside, or being $\mathcal O(h^\infty)$, on subsets
of $\overline T^*((-\varepsilon,T(h)+\varepsilon)\times X)$ independent of $t$.

We now discuss restrictions to space slices.
Assume that $u(h)\in\mathcal D'((-\varepsilon,T(h)+\varepsilon)\times X)$
is $h$-tempered uniformly in $t$ and moreover, $\widetilde\WFh(u)$ does not intersect the spatial
fiber infinity $\{\xi=0,\tau=\infty\}$. Then $u$ (as well as all its derivatives in $t$)
is a smooth function of $t$ with values in $\mathcal D'(X)$,
$u(t)$ is $h$-tempered uniformly in $t\in [0,T(h)]$, and
$$
\WFh(u(t))\subset\{(x,\xi)\mid \exists\tau: (x,\xi,\tau)\in \widetilde\WFh(u)\},
$$
uniformly in $t\in [0,T(h)]$. One can see this using~\eqref{e:wf-fourier}
and the formula for the Fourier transform of the restriction $w$
of $v\in\mathcal S'(\mathbb R^{n+1})$ to the hypersurface $\{t=0\}$:
$$
\mathcal F_h w(\xi)=(2\pi h)^{-1/2}\int_{\mathbb R} \mathcal F_hv(\xi,\tau)\,d\tau.
$$

\subsection{General assumptions}
  \label{s:kdsu-assumptions}

\newcounter{wassumptions}

In this section, we study Lorentzian metrics whose space slice is noncompact,
and define $r$-normal hyperbolicity and the dynamical quantities $\nu_{\min},\nu_{\max}$
in this case.

\smallsection{Geometric assumptions}
We assume that:
\switch{\begin{enumerate}}{\begin{enumerate}[(1)]}
\item \label{w:basic}
$(\widetilde X_0,\tilde g)$ is an $n+1$ dimensional Lorentzian manifold
of signature $(1,n)$, and $\widetilde X_0=\mathbb R_t\times X_0$, where
$X_0$, the space slice, is a manifold without boundary;
\item \label{w:stationary}
the metric $\tilde g$ is stationary in the sense that its coefficients do not depend
on $t$, or equivalently, $\partial_t$ is a Killing field;
\item \label{w:spacelike}
the space slices $\{t=\const\}$ are spacelike, or equivalently, the covector
$dt$ is timelike with respect to the dual metric $\tilde g^{-1}$ on $T^*\widetilde X_0$;
\setcounter{wassumptions}{\value{enumi}}
\end{enumerate}
The (nonsemiclassical) principal symbol of the d'Alembert--Beltrami operator $\Box_{\tilde g}$
(without the negative sign), denoted by $\tilde p(\tilde x,\tilde\xi)$, is
\begin{equation}
  \label{e:tilde-p}
\tilde p(\tilde x,\tilde\xi)=-\tilde g^{-1}_{\tilde x}(\tilde\xi,\tilde\xi),
\end{equation}
here $\tilde x=(t,x)$ denotes a point in $\widetilde X_0$ and
$\tilde\xi=(\tau,\xi)$ a covector in $T^*_{\tilde x}\widetilde X_0$.
The Hamiltonian flow of $\tilde p$ is the (rescaled) geodesic flow on $T^*_{\tilde x}\widetilde X_0$;
we are in particular interested in nontrivial lightlike geodesics,
i.e. the flow lines of $H_{\tilde p}$ on the set $\{\tilde p=0\}\setminus 0$,
where $0$ denotes the zero section.

Note that we do not assume that the vector field $\partial_t$ is timelike,
since this is false inside the ergoregion for rotating black holes.
Because of this, the intersections of the sets $\{\tau=\const\}$, invariant
under the geodesic flow, with the energy surface $\{\tilde p=0\}$ need not be
compact in the $\xi$ direction, and it is possible that $\xi$ will blow up in finite time
along a flow line of $H_{\tilde p}$, while $x$ stays in a compact subset of $X_0$.%
\footnote{The simplest example of such behavior is $\tilde p=x\xi^2+2\xi\tau-\tau^2$,
considering the geodesic starting at $x=t=\tau=0,\xi=1$.}
We consider instead the rescaled flow
\begin{equation}
  \label{e:rescaled-flow}
\tilde\varphi^s:=\exp(sH_{\tilde p}/\partial_\tau\tilde p)\quad\text{on }
\{\tilde p=0\}\setminus 0.
\end{equation}
Here $\partial_\tau\tilde p(\tilde x,\tilde\xi)
=-2\tilde g^{-1}_{\tilde x}(\tilde \xi,dt)$ never vanishes on $\{\tilde p=0\}\setminus 0$ by
assumption~\eqref{w:spacelike}.
Since $H_{\tilde p}t=\partial_\tau\tilde p$, the variable $t$ grows linearly with unit rate
along the flow $\tilde\varphi^s$.
The flow lines of~\eqref{e:rescaled-flow} exist for all $s$ as long as $x$ stays
in a compact subset of $X_0$. The flow is homogeneous, which makes
it possible to define it on the cosphere bundle $S^*\widetilde X_0$, which
is the quotient of $T^*\widetilde X_0\setminus 0$ by the action of dilations.
Finally, the flow preserves the restriction of the symplectic form to the tangent
bundle of $\{\tilde p=0\}$.

We next assume the existence of a `defining function of infinity' $\mu$
on the space slice with a concavity property:
\switch{\begin{enumerate}}{\begin{enumerate}[(1)]}
\setcounter{enumi}{\value{wassumptions}}
\item \label{w:convex}
there exists a function $\mu\in C^\infty(X_0)$
such that $\mu>0$ on $X_0$, for $\delta>0$ the set
\begin{equation}
  \label{e:X-delta}
X_\delta:=\{\mu>\delta\}\subset X_0
\end{equation}
is compactly contained in $X_0$, and there exists $\delta_0>0$ such that
for each flow line $\gamma(s)$ of~\eqref{e:rescaled-flow}, and with $\mu$
naturally defined on $T^*\widetilde X_0$,
\begin{equation}
  \label{e:concavity}
\mu(\gamma(s))<\delta_0,\
\partial_s\mu(\gamma(s))=0\quad\Longrightarrow\quad
\partial_s^2\mu(\gamma(s))<0.
\end{equation}
\setcounter{wassumptions}{\value{enumi}}
\end{enumerate}
We now define the trapped set:
\begin{defi}
  \label{d:time-trapping}
Let $\gamma(s)$ be a maximally extended flow line of~\eqref{e:rescaled-flow}.
We say that $\gamma(s)$ is trapped as $s\to +\infty$, if there exists
$\delta>0$ such that $\mu(\gamma(s))>\delta$ for all $s\geq 0$
(and as a consequence, $\gamma(s)$ exists for all $s\geq 0$).
Denote by $\widetilde\Gamma_-$ the union of all $\gamma$ trapped as $s\to +\infty$;
similarly, we define the union $\widetilde\Gamma_+$ of all $\gamma$ trapped as $s\to -\infty$.
Define the trapped set $\widetilde K:=\widetilde\Gamma_+\cap\widetilde\Gamma_-\subset \{\tilde p=0\}\setminus 0$.
\end{defi}
If $\mu(\gamma(s))<\delta_0$ and $\partial_s\mu(\gamma(s))\leq 0$ for some $s$, then
it follows from assumption~\eqref{w:convex} that $\gamma(s)$ is not trapped
as $s\to +\infty$. Also, if $\gamma(s)$ is not trapped as $s\to +\infty$, then
$\mu(\gamma(s))<\delta_0$ and $\partial_s\mu(\gamma(s))<0$ for $s>0$ large enough.
It follows that $\widetilde\Gamma_\pm$ are closed conic subsets of $\{\tilde p=0\}\setminus 0$,
and $\widetilde K\subset\{\mu\geq \delta_0\}$.

We next split the light cone $\{\tilde p=0\}\setminus 0$ into the sets
$\mathcal C_+$ and $\mathcal C_-$ of positively and negatively time oriented covectors:
\begin{equation}
  \label{e:time-orientation}
\mathcal C_\pm=\{\tilde p=0\}\cap \{\pm \partial_{\tau}\tilde p>0\}.
\end{equation}
Since $\partial_{\tau}\tilde p$ never vanishes on $\{\tilde p=0\}\setminus 0$
by assumption~\eqref{w:spacelike}, we have $\{\tilde p=0\}\setminus 0=\mathcal C_+\sqcup\mathcal C_-$.

We fix the sign of $\tau$ on the trapped set, in particular requiring
that $\widetilde K\subset \{\tau\neq 0\}$:
\switch{\begin{enumerate}}{\begin{enumerate}[(1)]}
\setcounter{enumi}{\value{wassumptions}}
\item \label{w:positive}
$\widetilde K\cap \mathcal C_\pm\subset \{\pm\tau <0\}$.
\setcounter{wassumptions}{\value{enumi}}
\end{enumerate}
\smallsection{Dynamical assumptions}
We now formulate the assumptions on the dynamical structure of the
flow~\eqref{e:rescaled-flow}. They are analogous to the assumptions of~\switch{\cite[\S5.1]{nhp}}{\S\ref{s:dynamics}}
and related to them in~\S\ref{s:kdsu-space} below. We start by requiring that
$\widetilde\Gamma_\pm$ are regular:
\switch{\begin{enumerate}}{\begin{enumerate}[(1)]}
\setcounter{enumi}{\value{wassumptions}}
\item \label{w:cr}
for a large constant $r$,
$\widetilde\Gamma_\pm$ are codimension 1 orientable $C^r$ submanifolds of $\{\tilde p=0\}\setminus 0$;
\item \label{w:symplectic}
$\widetilde\Gamma_\pm$ intersect transversely inside $\{\tilde p=0\}\setminus 0$,
and the intersections $\widetilde K\cap\{t=\const\}$ are symplectic submanifolds
of $T^*\widetilde X_0$.
\setcounter{wassumptions}{\value{enumi}}
\end{enumerate}
We next define a natural invariant decomposition of the tangent space to $\{\tilde p=0\}$
at $\widetilde K$. Let $(T\widetilde\Gamma_\pm)^{\perp}$ be the symplectic
complement of the tangent space to $\widetilde\Gamma_\pm$. Since $\widetilde\Gamma_\pm$
has codimension 2 and is contained in $\{\tilde p=0\}$,
$(T\widetilde\Gamma_\pm)^{\perp}$ is a two-dimensional vector subbundle
of $T(T^*\widetilde X_0)$ containing $H_{\tilde p}$. Since
$H_{\tilde p}t\neq 0$ on $\{\tilde p=0\}\setminus 0$, we can define the one-dimensional
vector subbundles of $T(T^*\widetilde X_0)$
\begin{equation}
  \label{e:tilde-v-pm}
\widetilde{\mathcal V}_\pm:=(T\widetilde\Gamma_\pm)^\perp\cap\{dt=0\}.
\end{equation}
Since $\widetilde\Gamma_\pm$ is a codimension 1 submanifold of $\{\tilde p=0\}$
and $H_{\tilde p}$ is tangent to $\widetilde\Gamma_\pm$, we see that
$\widetilde\Gamma_\pm$ is coisotropic and then $\widetilde{\mathcal V}_\pm$
are one-dimensional subbundles of $T\widetilde\Gamma_\pm$;
moreover, since $\partial_t\in T\widetilde\Gamma_\pm$, we find $\widetilde{\mathcal V}_\pm\subset\{d\tau=0\}$. Since
$\widetilde K\cap \{t=\const\}$ is symplectic, we have
\begin{equation}
  \label{e:tk-split}
T_{\widetilde K}\widetilde\Gamma_\pm=T\widetilde K\oplus \widetilde{\mathcal V}_\pm|_{\widetilde K},\quad
T_{\widetilde K}\tilde p^{-1}(0)=T\widetilde K\oplus \widetilde{\mathcal V}_-|_{\widetilde K}\oplus\widetilde{\mathcal V}_+|_{\widetilde K}.
\end{equation}
Since the flow
$\tilde\varphi^s$ from~\eqref{e:rescaled-flow} maps the space slice $\{t=t_0\}$ to $\{t=t_0+s\}$ and
$H_{\tilde p}$ is tangent to $T\widetilde\Gamma_\pm$, we see
that the splittings~\eqref{e:tk-split} are invariant under $\tilde\varphi^s$.

We now formulate the dynamical assumptions on the linearization of the flow $\tilde\varphi^s$
with respect to the splitting~\eqref{e:tk-split}.
Define the minimal expansion rate in the transverse direction $\nu_{\min}$
as the supremum of all $\nu$ for which there exists a constant $C$ such that
\begin{equation}
  \label{e:tilde-nu-min}
\sup_{\tilde\rho\in\widetilde K}\|d\tilde\varphi^{\mp s}(\tilde\rho)|_{\mathcal V_\pm}\|\leq Ce^{-\nu s},\quad
s\geq 0,
\end{equation}
with $\|\cdot\|$ denoting any smooth $t$-independent norm on the fibers of $T(T^*\widetilde X_0)$,
homogeneous of degree zero with respect to dilations on $T^*\widetilde X_0$.
Similarly, define $\nu_{\max}$ as the infimum of all $\nu$ for which these exists a constant $c>0$ such that
\begin{equation}
  \label{e:tilde-nu-max}
\inf_{\tilde\rho\in\widetilde K}\|d\tilde\varphi^{\mp s}(\tilde\rho)|_{\mathcal V_\pm}\|\geq ce^{-\nu s},\quad
s\geq 0.
\end{equation}
We now formulate the dynamical assumption of $r$-normal hyperbolicity:
\switch{\begin{enumerate}}{\begin{enumerate}[(1)]}
\setcounter{enumi}{\value{wassumptions}}
\item \label{w:nhp}
$\nu_{\min}>r\mu_{\max}$, where $\mu_{\max}$ is the maximal expansion rate of the flow along
$\widetilde K$, defined as the infimum of all $\nu$ for which there exists a constant $C$ such that
\begin{equation}
  \label{e:tilde-mu-max}
\sup_{\tilde\rho\in\widetilde K}\|d\tilde\varphi^s(\tilde\rho)|_{T\widetilde K}\|\leq Ce^{\nu s},\quad
s\in\mathbb R.
\end{equation}
\setcounter{wassumptions}{\value{enumi}}
\end{enumerate}
The large constant $r$ determines
how many terms we need to obtain in semiclassical expansions, and how many
derivatives of these terms need to exist~-- see\switch{~\cite{nhp}}{ Chapter~\ref{c:nhp}}. Theorem~\ref{t:kds-weyl}
simply needs $r$ to be large (in principle, depending on the dimension),
while Theorems~\ref{t:decay} and~\ref{t:res-dec} require $r$ to be large enough
depending on $N,T$. For exact Kerr(--de Sitter)
metrics, our assumptions are satisfied for all $r$, but a small perturbation
will satisfy them for some fixed large $r$ depending on the size of the perturbation.

\subsection{Reduction to the space slice}
  \label{s:kdsu-space}

We now put a Lorentzian manifold $(\widetilde X_0,\tilde g)$
satisfying assumptions of~\S\ref{s:kdsu-assumptions}
into the framework of\switch{~\cite{nhp}}{ Chapter~\ref{c:nhp}}. Consider the
stationary d'Alembert--Beltrami operator $P_{\tilde g}(\omega)$,
$\omega\in\mathbb C$, the second order semiclassical differential operator on the
space slice $X_0$ obtained by replacing $hD_t$ by $-\omega$ in the semiclassical d'Alembert--Beltrami
operator $h^2\Box_{\tilde g}$. The principal symbol
of $P_{\tilde g}(\omega)$ is given by
$$
\mathbf p(x,\xi;\omega)=\tilde p(t,x,-\omega,\xi),
$$
where $\tilde p$ is defined in~\eqref{e:tilde-p} and
the right-hand side does not depend on $t$. We will show that the operator
$P_{\tilde g}(\omega)$ satisfies a subset of the assumptions of~\switch{\cite[\S\S4.1, 5.1]{nhp}}{\S\S\ref{s:framework-assumptions},
\ref{s:dynamics}}.

First of all, we need to understand the solutions in $\omega$ to the equation
$\mathbf p=0$. Let
$$
p(x,\xi)\in C^\infty(T^*X_0\setminus 0)
$$
be the unique real solution $\omega$ to the equation $\mathbf p(x,\xi;\omega)=0$ such that
$(t,x,-\omega,\xi)\in \mathcal C_+$, with the positive time oriented light cone $\mathcal C_+$ defined in~\eqref{e:time-orientation}.
The existence and uniqueness of such solution follows from assumption~\eqref{w:spacelike} in~\S\ref{s:kdsu-assumptions},
and we also find from the definition of $\mathcal C_+$ that
\begin{equation}
  \label{e:rex}
\partial_{\omega}\mathbf p(x,\xi;p(x,\xi))<0,\quad
(x,\xi)\in T^*X_0\setminus 0.
\end{equation}
We can write $\mathcal C_+$ as the graph of $p$:
$$
\mathcal C_+=\{(t,x,-p(x,\xi),\xi)\mid t\in\mathbb R,\ (x,\xi)\in T^*X_0\setminus 0\}.
$$
The level sets of $p$ are not compact if $\partial_t$ is not timelike. To avoid dealing
with the fiber infinity, we use assumption~\eqref{w:positive} in~\S\ref{s:kdsu-assumptions}
to identify a bounded region in $T^*X_0$ invariant under the flow and containing the trapped set:
\begin{lemm}
  \label{l:W-construction}
There exists an open conic subset $\mathcal W\subset \mathcal C_+$, independent of $t$,
such that $\widetilde K\cap \mathcal C_+\subset \mathcal W$, the closure of $\mathcal W$
in $\mathcal C_+$ is contained in $\{\tau<0\}$, and $\mathcal W$ is invariant
under the flow~\eqref{e:rescaled-flow}.
\end{lemm}
\begin{proof}
Consider a conic neighborhood $\mathcal W_0$ of $\widetilde K\cap \mathcal C_+$
in $\mathcal C_+$ independent of $t$ and
such that the closure of $\mathcal W_0$ is contained in $\{\mu>\delta_0/2\}\cap
\{\tau<0\}$; this is possible by assumption~\eqref{w:positive}
and since $\widetilde K$ is contained in $\{\mu\geq\delta_0\}$.
Let $\mathcal W\subset \mathcal C_+$ be the union of all
maximally extended flow lines of~\eqref{e:rescaled-flow} passing through
$\mathcal W_0$.
Then $\mathcal W$ is an open conic subset of $\mathcal C_+$
containing $\widetilde K\cap\mathcal C_+$ and invariant under the flow~\eqref{e:rescaled-flow}.
It remains to show that each point $(\tilde x,\tilde\xi)\in \mathcal C_+\cap \{\tau\geq 0\}$
has a neighborhood that does not intersect $\mathcal W$. To see this, note that
the corresponding trajectory $\gamma(s)$ of~\eqref{e:rescaled-flow} does not
lie in $\widetilde\Gamma_+\cup\widetilde\Gamma_-$ (as otherwise, the projection of $\gamma(s)$
onto the cosphere bundle
would converge to $\widetilde K$ as $s\to+\infty$ or $s\to -\infty$,
by~\switch{\cite[Lemma~4.1]{nhp}}{Lemma~\ref{l:the-flow}}; it remains to use assumption~\eqref{w:positive}
and the fact that $\tau$ is constant on $\gamma(s)$).
We then see that
$\gamma(s)$ escapes for both $s\to+\infty$ and $s\to-\infty$ and
does not intersect the closure of $\mathcal W_0$ and same is true
for nearby trajectories; therefore, a neighborhood of $(\tilde x,\tilde\xi)$ does not
intersect $\mathcal W$.
\end{proof}
Arguing similarly (using an open conic subset $\mathcal W'_0$ of $\mathcal C_+$ such
that $\overline{\mathcal W_0}\subset \mathcal W'_0$
and $\overline{\mathcal W'_0}\subset \{\mu>\delta_0/2\}\cap \{\tau<0\}$),
we construct an open conic subset $\mathcal W'$ of $\mathcal C_+$ independent of $t$ and
such that
$$
\widetilde K\cap \mathcal C_+\subset \mathcal W,\quad
\overline{\mathcal W}\subset \mathcal W',\quad
\overline{\mathcal W'}\subset \{\tau<0\},
$$
and $\mathcal W,\mathcal W'$ are invariant under the flow~\eqref{e:rescaled-flow}. Now, take
small $\delta_1>0$ and define
\begin{equation}
  \label{e:fUnz}
\begin{aligned}
\widetilde{\mathcal U}&:=\mathcal C_+\cap \{|1+\tau|<\delta_1\}\cap \mathcal W\cap \{\mu>\delta_1\},\\
\widetilde{\mathcal U}'&:=\mathcal C_+\cap \{|1+\tau|<2\delta_1\}\cap \mathcal W'\cap \{\mu>\delta_1/2\}.
\end{aligned}
\end{equation}
Then $\widetilde{\mathcal U},\widetilde{\mathcal U}'$ are open subsets of $\mathcal C_+$ convex under the flow~\eqref{e:rescaled-flow},
$\widetilde K\cap \{|1+\tau|<\delta_1\}\subset\widetilde{\mathcal U}$
(note that $\widetilde K\cap \{\tau<0\}\subset \mathcal C_+$
by assumption~\eqref{w:positive}),
and the closure of $\widetilde{\mathcal U}$ is contained in $\widetilde{\mathcal U}'$.
Moreover, the projections of $\mathcal U,\mathcal U'$ onto the $(x,\tau,\xi)$ variables are bounded
because $\mathcal W,\mathcal W'$ are conic and $\overline{\mathcal W},\overline{\mathcal W'}\subset \{\tau\neq 0\}$.

Let $\mathcal U\Subset\mathcal U'\Subset T^* X_0$ be the projections of $\widetilde{\mathcal U},\widetilde{\mathcal U}'$
onto the $(x,\xi)$ variables, so that
$$
\widetilde{\mathcal U}=\{(t,x,-p(x,\xi),\xi)\mid t\in\mathbb R,\ (x,\xi)\in\mathcal U\},
$$
and similarly for $\mathcal U'$. Note that $\mathcal U\subset \{|p-1|<\delta_1\}$
and $\mathcal U'\subset \{|p-1|<2\delta_1\}$.
Since $\mathcal U'$ is bounded, and by~\eqref{e:rex}, for $\delta_1>0$ small enough and $(x,\xi)\in\mathcal U'$,
$p(x,\xi)$ is the only solution to the equation
$\mathbf p(x,\xi;\omega)=0$ in $\{\omega\in\mathbb C\mid |\omega-1|<2\delta_1\}$.

We now study the Hamiltonian flow of $p$. Since
$$
\partial_{x,\xi}p(x,\xi)=-{\partial_{x,\xi} \mathbf p(x,\xi,p(x,\xi))\over \partial_\omega \mathbf p(x,\xi,p(x,\xi))},
$$
and for each $t$,
$$
-\partial_\omega\mathbf p(x,\xi,p(x,\xi))
=\partial_\tau\tilde p(t,x,-p(x,\xi),\xi),
$$
we see that the flow of $H_p$ is the projection of the rescaled geodesic flow~\eqref{e:rescaled-flow}
on $\mathcal C_+$: for $(x,\xi)\in T^*X_0\setminus 0$,
\begin{equation}
  \label{e:rescaled-flow-2}
\tilde\varphi^s(t,x,-p(x,\xi),\xi)=(t+s,x(s),-p(x,\xi),\xi(s)),\quad
(x(s),\xi(s))=e^{sH_p}(x,\xi).
\end{equation}
We now verify some of the assumptions of~\switch{\cite[\S4.1]{nhp}}{\S\ref{s:framework-assumptions}}. We let $X$ be an $n$-dimensional manifold
containing $X_0$ (for the Kerr--de Sitter metric it is constructed in~\S\ref{s:kdsu-resonances})
and consider the volume form $d\Vol$ on $X_0$ related to the volume form $d\,\widetilde{\Vol}$
on $\widetilde X_0$ generated by $\tilde g$ by the formula
$d\,\widetilde{\Vol}=dt\wedge d\Vol$. The operator $P_{\tilde g}(\omega)$
is a semiclassical pseudodifferential operator depending holomorphically
on $\omega\in\Omega:=\{|\omega-1|<2\delta_1\}$ and $\mathbf p$ is its semiclassical
principal symbol. We do not specify the spaces $\mathcal H_1,\mathcal H_2$ here
and do not establish any mapping or Fredholm properties of $P_{\tilde g}(\omega)$;
for our specific applications it is done in~\S\ref{s:kdsu-resonances}.
Except for these mapping properties, the assumptions~(1), (2), and~(5)--(9) of~\switch{\cite[\S4.1]{nhp}}{\S\ref{s:framework-assumptions}}
are satisfied, with $\mathcal U,\mathcal U'$ defined above,
$[\alpha_0,\alpha_1]:=[1-\delta_1/2,1+\delta_1/2]$, and the incoming/outgoing
tails $\Gamma_\pm$ on the space slice given by
(for each $t$)
\begin{equation}
  \label{e:gamma-pm-new}
\Gamma_\pm=\{(x,\xi)\mid (t,x,-p(x,\xi),\xi)\in \widetilde\Gamma_\pm \cap \{|1+\tau|\leq \delta_1\}\cap \overline{\mathcal W}
\cap \{\mu\geq\delta_1\}\},
\end{equation}
and similarly for the trapped set $K=\Gamma_+\cap\Gamma_-$.

Finally, the dynamical assumptions of~\switch{\cite[\S5.1]{nhp}}{\S\ref{s:dynamics}}
are also satisfied, as follows directly from \eqref{e:rescaled-flow-2}
and the dynamical assumptions of~\S\ref{s:kdsu-assumptions}. Note that the subbundles
$\mathcal V_\pm$ of $T\Gamma_\pm$ defined in~\switch{\cite[\S5.1]{nhp}}{\S\ref{s:dynamics}} coincide with the subbundles
$\widetilde{\mathcal V}_\pm$ of $T\widetilde\Gamma_\pm$ defined in~\S\ref{e:tilde-v-pm}
under the identification $T_{(x,\xi)}(T^*X_0)\simeq T_{(t,x,-p(x,\xi),\xi)}(T^*\widetilde X_0)\cap \{dt=d\tau=0\}$,
and the expansion rates $\nu_{\min},\nu_{\max},\mu_{\max}$ defined in~\eqref{e:tilde-nu-min}--\eqref{e:tilde-mu-max}
coincide with those defined in~\switch{\cite[(5.1)--(5.3)]{nhp}}{\eqref{e:nu-min}--\eqref{e:mu-max}}.

To relate the constants for the Weyl laws in Theorem~\ref{t:kds-weyl} and\switch{~\cite[Theorem~2]{nhp}}{ Theorem~\ref{t:weyl-law}},
we note that for $[a,b]\subset (1-\delta_1/2,1+\delta_1/2)$,
$$
\Vol_\sigma(K\cap p^{-1}[a,b])=\Vol_{\tilde\sigma}(\widetilde K\cap \{a\leq -\tau\leq b\}\cap\{t=\const\}).
$$
Here $\Vol_\sigma$ and $\Vol_{\tilde\sigma}$ stand for symplectic volume forms
of order $2n-2$ on $T^*X_0$ and $T^*\widetilde X_0$, respectively. The constant $c_{\widetilde K}$
from Theorem~\ref{t:kds-weyl} is then given by
\begin{equation}
  \label{e:c-k-tilde}
c_{\widetilde K}=\Vol_{\tilde\sigma}(\widetilde K\cap \{0\leq\tau\leq 1\}\cap \{t=\const\}).
\end{equation}

\subsection{Applications to linear waves}
  \label{s:kdsu-decay}

In this section, we apply the results of~\switch{\cite{nhp}}{Chapter~\ref{c:nhp}} to understand the decay properties
of linear waves; Theorem~\ref{t:internal-waves} below forms the base for the proofs
of Theorems~\ref{t:decay} and~\ref{t:res-dec} in~\S\ref{s:kdsu-waves}.

Consider a family of approximate solutions $u(h)\in\mathcal D'((-1,T(h)+1)_t\times X_0)$ to the wave equation
\begin{equation}
  \label{e:wewe}
h^2\Box_{\tilde g}u(h)=\mathcal O(h^\infty)_{C^\infty}.
\end{equation}
Here $h\ll 1$ is the semiclassical parameter and $T(h)>0$ depends on $h$ (for our particular
application, $T(h)=T\log(1/h)$ for some constant $T$). We assume that
$u$ is $h$-tempered uniformly in $t$, as defined in~\S\ref{s:kdsu-prelims}. Then
by the elliptic estimate (see for instance\switch{~\cite[Proposition~3.2]{nhp}}{ Proposition~\ref{l:elliptic}}), $u$ is microlocalized on the
light cone:
\begin{equation}
  \label{e:woof}
\widetilde\WFh(u)\subset \{\tilde p=0\},
\end{equation}
where $\widetilde\WFh(u)$ is defined in~\S\ref{s:kdsu-prelims}.
By the restriction statement in~\S\ref{s:kdsu-prelims}, $u$ is
a smooth function of $t$ with values in $h$-tempered distributions on $X_0$.
Moreover, we obtain for $0<\delta_1<\delta_2$ small enough and each $t_0\in [0,T(h)]$,
\begin{equation}
  \label{e:equivl2}
\begin{gathered}
\|u(t_0)\|_{H^1_x(X_{\delta_2})}\leq C\|u\|_{H^1_{t,x}([t_0-1,t_0+1]\times X_{\delta_1})}+\mathcal O(h^\infty),\\
\|u\|_{L^2_t[t_0-1,t_0+1]L^2_x(X_{\delta_1})}\leq C\|u\|_{L^\infty_t[t_0-1,t_0+1]L^2_x(X_{\delta_1})}+\mathcal O(h^\infty).
\end{gathered}
\end{equation}
The second of these inequalities is trivial; the first one is done by applying
the standard energy estimate for the wave equation to the function
$\chi(t-t_0)u$, with $\chi\in C_0^\infty(-\epsilon,\epsilon)$ equal to 1 near $0$
and $\epsilon>0$ small depending on $\delta_1,\delta_2$.

We furthermore restrict ourselves to the following class of outgoing solutions, see Figure~\ref{f:outgoing}(a):
\begin{defi}
  \label{d:we-outgoing}
Fix small $\delta_1>0$.
A solution $u$ to~\eqref{e:wewe}, $h$-tempered uniformly in $t\in (-1,T(h)+1)$,
is called outgoing, if its projected wavefront set $\widetilde\WFh(u)$, defined in~\S\ref{s:kdsu-prelims},
satisfies (for $\widetilde{\mathcal U}$ defined in~\eqref{e:fUnz})
\begin{align}
\label{e:we-outgoing-1}
\widetilde\WFh(u)\cap \{\mu>\delta_1\}&\subset \widetilde{\mathcal U}\cap \{|\tau+1|<\delta_1/4\},\\
\label{e:we-outgoing-2}
\widetilde\WFh(u)\cap \{\delta_1\leq\mu\leq 2\delta_1\}&\subset \{H_{\tilde p}\mu\leq 0\}.
\end{align}
\end{defi}
%
\begin{figure}
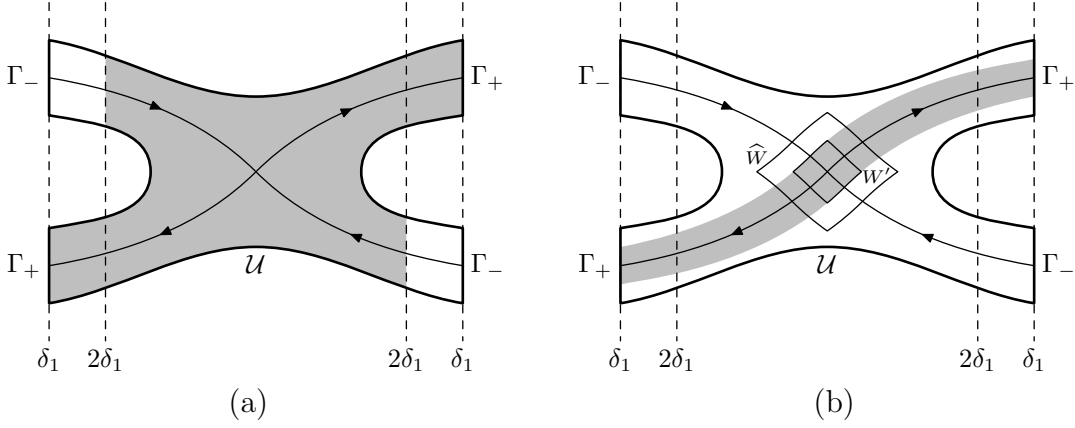

\includegraphics{kdsu.10}
\qquad
\includegraphics{kdsu.11}
\hbox to\hsize{\hss (a) \hss\hss (b)\hss}
\caption[The phase space picture of the flow, showing shaded
$\WFh(\mathbf u(t))$ for (a) all $t$ and (b) all $t\geq t_1$,
for $u$ satisfying Definition~\ref{d:we-outgoing}.]%
{The phase space picture of the flow, showing shaded
$\WFh(\mathbf u(t))$ for (a) all $t$ and (b) all $t\geq t_1$,
for $u$ satisfying Definition~\ref{d:we-outgoing}. The horizontal axis corresponds
to $\mu$.}
\label{f:outgoing}
\end{figure}
The main result of this section is
\begin{theo}\label{t:internal-waves}
Fix $T,N,\varepsilon>0$ and let the assumptions of~\S\ref{s:kdsu-assumptions} hold,
including $r$-normal hyperbolicity with $r$ large depending on $T,N$.
Assume that $u$ is an outgoing solution to~\eqref{e:wewe}, for
$t\in (-1,T\log(1/h)+1)$, and $\|u(t)\|_{H^{-N}_h(X_{\delta_1/2})}=\mathcal O(h^{-N})$
uniformly in $t$.
Then for $t_0$ large enough and independent of $h$, we can write
$$
u(t,x)=u_\Pi(t,x)+u_R(t,x),\quad
t_0\leq t\leq T\log(1/h),
$$
such that $h^2\Box_{\tilde g}u_\Pi,h^2\Box_{\tilde g}u_R$ are $\mathcal O(h^N)_{H^N_h}$
on $X_{\delta_1}$ and, with $\|\cdot\|_{\mathcal E}$
defined in~\eqref{e:kdsu-energy},
\begin{align}
  \label{e:ini-1}
\|u_\Pi(t_0)\|_{\mathcal E}&\leq Ch^{-1/2}\|u(0)\|_{\mathcal E}+\mathcal O(h^N),\\
  \label{e:ini-2}
\|u_\Pi(t)\|_{\mathcal E}&\leq Ce^{-(\nu_{\min}-\varepsilon)t/2}\|u_\Pi(t_0)\|_{\mathcal E}+\mathcal O(h^N),\\
  \label{e:ini-3}
\|u_\Pi(t)\|_{\mathcal E}&\geq C^{-1}e^{-(\nu_{\max}+\varepsilon)t/2}\|u_\Pi(t_0)\|_{\mathcal E}-\mathcal O(h^N),\\
  \label{e:ini-4}
\|u_R(t)\|_{\mathcal E}&\leq Ch^{-1}e^{-(\nu_{\min}-\varepsilon)t}\|u(0)\|_{\mathcal E}+\mathcal O(h^N),\\
  \label{e:ini-5}
\|u(t)\|_{\mathcal E}&\leq Ce^{\varepsilon t}\|u(0)\|_{\mathcal E}+\mathcal O(h^N),
\end{align}
all uniformly in $t\in [t_0,T\log(1/h)]$.
\end{theo}
For the proof, we assume that the metric is $r$-normally hyperbolic for all $r$,
and prove the bounds for all $T,N$ (so that $\mathcal O(h^N)$ becomes $\mathcal O(h^\infty)$);
since semiclassical arguments require finitely many derivatives to work, the statement
will be true for $r$ large depending on $T$ and $N$.

We first recall the factorization of\switch{~\cite[Lemma~4.3]{nhp}}{ Lemma~\ref{l:resolution}}:
\begin{equation}
  \label{e:fac-1}
P_{\tilde g}(\omega)=\mathcal S(\omega)(P-\omega)\mathcal S(\omega)+\mathcal O(h^\infty)\quad
\text{microlocally near }\mathcal U,
\end{equation}
where $\mathcal S(\omega)$ is a family of pseudodifferential operators elliptic near $\mathcal U$,
and such that $\mathcal S(\omega)^*=\mathcal S(\omega)$ for $\omega\in\mathbb R$, and
$P$ is a self-adjoint pseudodifferential operator, moreover we assume that it is compactly
supported and compactly microlocalized.
If we define the self-adjoint pseudodifferential operator $\widetilde{\mathcal S}$ on $\widetilde X_0$
by replacing $\omega$ by $-hD_t$ in $\mathcal S(\omega)$, then we get
\begin{equation}
\label{e:fac-2}
h^2\Box_{\tilde g}=\widetilde{\mathcal S}(hD_t+P)\widetilde{\mathcal S}+\mathcal O(h^\infty)\quad
\text{microlocally near }\widetilde{\mathcal U}.
\end{equation}
We define
$$
\mathbf u(t):=(\widetilde{\mathcal S}u)(t),\quad
0\leq t\leq T\log(1/h),
$$
note that $\mathbf u(t)$ and its $t$-derivatives are bounded
uniformly in $t$ with values in $h$-tempered distributions
on $X_0$ by the discussion of restrictions to space slices in~\S\ref{s:kdsu-prelims}
and by~\eqref{e:woof}. We have by~\eqref{e:wewe}, \eqref{e:we-outgoing-1}, \eqref{e:we-outgoing-2}, and~\eqref{e:fac-2},
\begin{gather}
  \label{e:kentucky-1}
(hD_t+P)\mathbf u(t)=\mathcal O(h^\infty)\quad\text{microlocally near }X_{\delta_1},\\
  \label{e:kentucky-1.5}
\WFh(\mathbf u(t))\cap X_{\delta_1}\subset \{|p-1|<\delta_1/4\},\\
  \label{e:kentucky-1.75}
\WFh(\mathbf u(t))\cap \{\delta_1\leq \mu\leq 2\delta_1\}\subset \{H_p\mu\leq 0\},
\end{gather}
uniformly in $t\in [0, T\log(1/h)]$.

We next use the construction of\switch{~\cite[Lemma~5.1]{nhp}}{ Lemma~\ref{l:phi-pm}},
which (combined with the homogeneity of the flow)
gives functions $\varphi_\pm$ defined in a conic neighborhood of $K$ in $T^*X_0$,
such that $\Gamma_\pm=\{\varphi_\pm=0\}$ in this neighborhood, $\varphi_\pm$
are homogeneous of degree zero, and
\begin{equation}
H_p\varphi_\pm=\mp c_\pm\varphi_\pm,\quad
\nu_{\min}-\varepsilon<c_\pm<\nu_{\max}+\varepsilon,
\end{equation}
where $c_\pm$ are some smooth functions on the domain of $\varphi_\pm$. Then
for small $\delta>0$,
$$
U_\delta:=\{|\varphi_+|\leq\delta,\ |\varphi_-|\leq\delta\}
$$
is a small closed conic neighborhood of $K$ in $T^*X_0\setminus 0$.

We now fix $\delta$ small enough so
that\switch{~\cite[Theorem~3 in~\S7.1 and Proposition~7.1]{nhp}}{ Theorem~\ref{t:our-Pi} in~\S\ref{s:construction-1}
and Proposition~\ref{l:ideals}}
apply, giving a Fourier integral operator $\Pi\in I_{\comp}(\Lambda^\circ)$ which satisfies the
equations
\begin{equation}
  \label{e:magic-Pi}
\Pi^2=\Pi+\mathcal O(h^\infty),\quad
[P,\Pi]=\mathcal O(h^\infty)
\end{equation}
microlocally near the set $\widehat W\times\widehat W$, with
\begin{equation}
  \label{e:kdsu-hat-w}
\widehat W:=U_\delta\cap \{|p-1|\leq \delta_1/2\}.
\end{equation}
Here $\Lambda^\circ\subset\Gamma_-\cap\Gamma_+$ is the canonical relation defined in\switch{~\cite[(5.12)]{nhp}}{~\eqref{e:the-Lambda}}. Also, we define
\begin{equation}
  \label{e:kdsu-w'}
W':=U_{\delta/2}\cap \{|p-1|\leq\delta_1/4\}.
\end{equation}

We now derive certain conditions on the microlocalization of~$u$ for large enough times,
see Figure~\ref{f:outgoing}(b) (compare with\switch{~\cite[Figure~5]{nhp}}{ Figure~\ref{f:reduction}}):
\begin{prop}
\label{l:kentucky}
For $t_1$ large enough independent of $h$, the function
$\mathbf u(t)$ satisfies
\begin{gather}
  \label{e:kentucky-2}
\WFh(\mathbf u(t))\cap \widehat W\subset\{|\varphi_+|<\delta/2\},\\
  \label{e:kentucky-2.5}
\WFh(\mathbf u(t))\cap \Gamma_-\subset W',
\end{gather}
uniformly in $t\in [t_1,T\log(1/h)]$.
\end{prop}
\begin{proof}
Consider $(x,\xi)\in \WFh(\mathbf u(t))\cap X_{2\delta_1}$ for some $t\in [t_1,T\log(1/h)]$.
Put $\gamma(s)=e^{sH_p}(x,\xi)$. Then by propagation of singularities
(see for example\switch{~\cite[Proposition~3.4]{nhp}}{ Proposition~\ref{l:microhyperbolic}}) for the equation~\eqref{e:kentucky-1},
we see that either there exists $s_0\in [-t_1,0]$ such that
$\gamma(s_0)\in\{\delta_1\leq \mu\leq 2\delta_1\}\cap \WFh(u(t+s_0))$,
or $\gamma(s)\in X_{2\delta_1}$ for all $s\in [-t_1,0]$.
However, in the first of these two cases, by~\eqref{e:kentucky-1.75}
we have $\gamma(s_0)\in \{\mu\leq 2\delta_1\}\cap \{H_p\mu\leq 0\}$,
which implies that $\gamma(0)\in \{\mu\leq 2\delta_1\}$ by assumption~\eqref{w:convex}
in~\S\ref{s:kdsu-assumptions}, a contradiction. Therefore,
$$
e^{tH_p}(x,\xi)\in X_{2\delta_1},\quad
t\in [-t_1,0].
$$
It remains to note that for $t_1$ large enough,
$$
\begin{gathered}
e^{-t_1H_p}(\widehat W\cap \{|\varphi_+|\geq \delta/2\})\cap X_{2\delta_1}=\emptyset;\\
e^{t_1H_p}(\Gamma_-\cap \{|p-1|<\delta_1/4\}\cap X_{2\delta_1})\subset W';
\end{gathered}
$$
the first of these statements follows from the fact that $\widehat W\cap \{|\varphi_+|\geq\delta/2\}$
is a compact set not intersecting $\Gamma_+$, and the second one, from\switch{~\cite[Lemma~4.1]{nhp}}{ Lemma~\ref{l:the-flow}}.
\end{proof}
By~\eqref{e:kentucky-1}, \eqref{e:magic-Pi}, and~\eqref{e:kentucky-2.5}, and since $\WFh(\Pi)\subset\Gamma_-\times\Gamma_+$
we have uniformly in $t\in [t_1,T\log(1/h)]$,
\begin{equation}
  \label{e:pier}
(hD_t+P)\Pi\mathbf u(t)=\mathcal O(h^\infty)\quad\text{microlocally near }
\widehat W.
\end{equation}
By\switch{~\cite[Proposition~6.1 and~\S6.2]{nhp}}{ Proposition~\ref{l:bund} and \S\ref{s:general}}, we have
\begin{equation}
  \label{e:inj-1}
\|\Pi \mathbf u(t)\|_{L^2}\leq Ch^{-1/2}\|\mathbf u(t)\|_{L^2}.
\end{equation}
We now use the methods of~\switch{\cite[\S8]{nhp}}{\S\ref{s:resolvent-bounds}} to prove a microlocal
version of Theorem~\ref{t:internal-waves} near the trapped set:
\begin{prop}
  \label{l:internal-waves-1}
There exist compactly supported $A_0,A_1\in\Psi^{\comp}(X_0)$
microlocalized inside $\widehat W$, elliptic on $W'$, and such that
for $t\in [t_1,T\log(1/h)]$,
\begin{align}
  \label{e:inj-2}
\|A_0\Pi\mathbf u(t)\|_{L^2}&\leq Ce^{-(\nu_{\min}-\varepsilon)t/2}\|A_0\Pi\mathbf u(t_1)\|_{L^2}+\mathcal O(h^\infty),\\
  \label{e:inj-3}
\|A_0\Pi\mathbf u(t)\|_{L^2}&\geq C^{-1}e^{-(\nu_{\max}+\varepsilon)t/2}\|A_0\Pi \mathbf u(t_1)\|_{L^2}-\mathcal O(h^\infty),\\
  \label{e:inj-4}
\|A_1(1-\Pi)\mathbf u(t)\|_{L^2}&\leq Ch^{-1}e^{-(\nu_{\min}-\varepsilon)t}\|A_0\mathbf u(t_1)\|_{L^2}+\mathcal O(h^\infty),\\
  \label{e:inj-5}
\|A_1\mathbf u(t)\|_{L^2}&\leq Ce^{\varepsilon t}\|A_0\mathbf u(t_1)\|_{L^2}+\mathcal O(h^\infty).
\end{align}
\end{prop}
\begin{proof}
We will use the operators $\Theta_\pm,\Xi$ constructed in\switch{~\cite[Proposition~7.1]{nhp}}{ Proposition~\ref{l:ideals}}.
The microlocalization statements we make will be uniform in $t\in [t_1,T\log(1/h)]$.

We first prove~\eqref{e:inj-4}, following the proof of\switch{~\cite[Proposition~8.1]{nhp}}{ Proposition~\ref{l:estimate-kernel}}.
Put
$$
\mathbf v(t):=\Xi\mathbf u(t).
$$
Then similarly to~\switch{\cite[(8.14)]{nhp}}{\eqref{e:kernel-int-1.5}}, we find
$$
(1-\Pi)\mathbf u(t)=\Theta_-\mathbf v(t)+\mathcal O(h^\infty)\quad\text{microlocally near }\widehat W.
$$
By~\eqref{e:kentucky-1} and~\eqref{e:pier},
$$
(hD_t+P)(1-\Pi)\mathbf u(t)=\mathcal O(h^\infty)\quad\text{microlocally near }\widehat W.
$$
Similarly to\switch{~\cite[Proposition~8.3]{nhp}}{ Proposition~\ref{l:estimate-kernel-int}},
we use the commutation relation $[P,\Theta_-]=-ih\Theta_-Z_-+\mathcal O(h^\infty)$
together with propagation of singularities for the operator $\Theta_-$ to find
\begin{equation}
  \label{e:veq}
(hD_t+P-ihZ_-)\mathbf v(t)=\mathcal O(h^\infty)\quad\text{microlocally near }\widehat W.
\end{equation}
Here $Z_-\in\Psi^{\comp}(X_0)$ satisfies $\sigma(Z_-)=c_-$ near $\widehat W$.

Let $\mathcal X_-\in\Psi^{\comp}(X_0)$ be the operator used in~\switch{\cite[\S8.2]{nhp}}{\S\ref{s:estimate-kernel}},
satisfying $\WFh(\mathcal X_-)\Subset \widehat W$, $\sigma(\mathcal X_-)\geq 0$
everywhere, and $\sigma(\mathcal X_-)>0$ on $W'$.
Similarly to~\switch{\cite[(8.18)]{nhp}}{\eqref{e:pos-comm-1}}, we get
\begin{equation}
  \label{e:xeq}
{1\over 2}\partial_t\langle \mathcal X_- \mathbf v(t),\mathbf v(t)\rangle
+\langle \mathcal Y_- \mathbf v(t),\mathbf v(t)\rangle=\mathcal O(h^\infty),
\end{equation}
where
$$
\mathcal Y_-={1\over 2}(Z_-^*\mathcal X_-+\mathcal X_-Z_-)+{1\over 2ih}[P,\mathcal X_-]
$$
satisfies $\WFh(\mathcal Y_-)\Subset \widehat W$,
and similarly to~\switch{\cite[(8.19)]{nhp}}{\eqref{e:y-minus-positive}} we have
$$
\sigma(\mathcal Y_-)\geq (\nu_{\min}-\varepsilon)\sigma(\mathcal X_-)\quad\text{near }\WFh(\mathbf v(t)),
$$
and the inequality is strict on $W'$. Similarly to\switch{~\cite[Lemma~8.4]{nhp}}{ Lemma~\ref{l:kernel-garding}},
by sharp G\r arding inequality we get
\begin{equation}
  \label{e:fhtagn}
\langle(\mathcal Y_- -(\nu_{\min}-\varepsilon)\mathcal X_-)\mathbf v(t),\mathbf v(t)\rangle\geq
\|A_1\mathbf v(t)\|_{L^2}^2-Ch\|A'_0\mathbf v(t)\|_{L^2}^2-\mathcal O(h^\infty)
\end{equation}
for an appropriate choice of $A_1$ and some
$A'_0\in\Psi^{\comp}(X_0)$ microlocalized inside $\widehat W$.
Also similarly to\switch{~\cite[Lemma~8.4]{nhp}}{ Lemma~\ref{l:kernel-garding}}, by propagation of singularities
for the equation~\eqref{e:veq} we get for $t_1$ large enough,
\begin{gather}
  \label{e:jessica2}
\|A'_0\mathbf v(t)\|_{L^2}^2\leq C\|A_0\mathbf v(t_1)\|_{L^2}^2+\mathcal O(h^\infty),\quad
t\in [t_1,2t_1],\\
  \label{e:jessica}
\|A'_0\mathbf v(t)\|_{L^2}^2\leq C\|A_1\mathbf v(t-t_1)\|_{L^2}^2+\mathcal O(h^\infty),\quad
t\geq 2t_1,
\end{gather}
for an appropriate choice of $A_0$.
By~\eqref{e:xeq} and~\eqref{e:fhtagn}, we see that
$$
\begin{gathered}
\langle\mathcal X_-\mathbf v(t),\mathbf v(t)\rangle
\leq Ce^{-2(\nu_{\min}-\varepsilon)t}\langle \mathcal X_-\mathbf v(t_1),\mathbf v(t_1)\rangle
-C^{-1}\int_{t_1}^t e^{-2(\nu_{\min}-\varepsilon)(t-s)}\|A_1\mathbf v(s)\|_{L^2}^2\,ds\\
+Ch\int_{t_1}^t e^{-2(\nu_{\min}-\varepsilon)(t-s)}\|A'_0\mathbf v(s)\|_{L^2}^2\,ds
+\mathcal O(h^\infty).
\end{gathered}
$$
Breaking the second integral on the right-hand side in two pieces and estimating
each of them separately by~\eqref{e:jessica2} and~\eqref{e:jessica}, we get
for an appropriate choice of $A_0$,
$$
\langle\mathcal X_-\mathbf v(t),\mathbf v(t)\rangle
\leq Ce^{-2(\nu_{\min}-\varepsilon)t}\|A_0\mathbf v(t_1)\|_{L^2}^2+\mathcal O(h^\infty).
$$
We can moreover assume that $\mathcal X_-$ has the form
$A_1^*A_1+\mathcal X_1^*\mathcal X_1+\mathcal O(h^\infty)$ for some
pseudodifferential operator $\mathcal X_1$; this can be arranged
since $\sigma(\mathcal X_-)>0$ on $\WFh(A_1)$ and
the argument of~\switch{\cite[\S8.2]{nhp}}{\S\ref{s:estimate-kernel}} only depends on the
principal symbol of $\mathcal X_-$, which can be taken to be
the square of a smooth function. Then
$\|A_1\mathbf v(t)\|_{L^2}^2\leq \langle\mathcal X_-\mathbf v(t),\mathbf v(t)\rangle+\mathcal O(h^\infty)$
and we get
\begin{equation}
  \label{e:inj-4.1}
\|A_1\mathbf v(t)\|_{L^2}\leq Ce^{-(\nu_{\min}-\varepsilon)t}\|A_0\mathbf v(t_1)\|_{L^2}
+\mathcal O(h^\infty).
\end{equation}
To prove~\eqref{e:inj-4}, it remains to note that $(1-\Pi)\mathbf u(t)=\mathbf v(t)+\mathcal O(h^\infty)$
microlocally near $\widehat W$ and $\|\mathbf v(t_1)\|_{L^2}\leq Ch^{-1}\|\mathbf u(t_1)\|_{L^2}$
by part~1 of\switch{~\cite[Proposition~6.13]{nhp}}{ Proposition~\ref{l:Xi-0}}.

To prove~\eqref{e:inj-5}, we argue similarly to~\eqref{e:xeq}, but use the equation~\eqref{e:kentucky-1}
instead of~\eqref{e:veq}. We get
$$
{1\over 2}\partial_t\langle\mathcal X_-\mathbf u(t),\mathbf u(t)\rangle
+\langle\mathcal Y'_-\mathbf u(t),\mathbf u(t)\rangle=\mathcal O(h^\infty),
$$
where
$$
\mathcal Y'_-={1\over 2ih}[P,\mathcal X_-]
$$
satisfies $\WFh(\mathcal Y'_-)\Subset \widehat W$ and
$$
\sigma(\mathcal Y'_-)\geq -\varepsilon\sigma(\mathcal X_-)\quad\text{near }\WFh(\mathbf u(t)),
$$
and the inequality is strict on $W'$.
The remainder of the proof of~\eqref{e:inj-5} proceeds exactly as the proof
of~\eqref{e:inj-4.1}.

Finally, we prove~\eqref{e:inj-2} and~\eqref{e:inj-3}, following the proof of\switch{~\cite[Proposition~8.2]{nhp}}{ Proposition~\ref{l:estimate-image}}.
Let $\mathcal X_+\in\Psi^{\comp}(X_0)$ be the operator defined in~\switch{\cite[\S8.3]{nhp}}{\S\ref{s:estimate-image}},
satisfying in particular $\WFh(\mathcal X_+)\Subset \widehat W$,
$\sigma(\mathcal X_+)\geq 0$ everywhere, and $\sigma(\mathcal X_+)>0$ on $W'$.
Similarly to~\switch{\cite[(8.33)]{nhp}}{\eqref{e:image-ultimate}}, we get from~\eqref{e:pier} that
for an appropriate choice of $A_0$,
\begin{equation}
  \label{e:bear}
{1\over 2}\partial_t\langle\mathcal X_+\Pi\mathbf u(t),\Pi\mathbf u(t)\rangle
+\langle\mathcal Z_+\Pi\mathbf u(t),\Pi\mathbf u(t)\rangle=\mathcal O(h)\|A_0\Pi\mathbf u(t)\|_{L^2}^2
+\mathcal O(h^\infty),
\end{equation}
where $\mathcal Z_+\in\Psi^{\comp}(X_0)$, $\WFh(\mathcal Z_+)\Subset\widehat W$,
$$
{\nu_{\min}-\varepsilon\over 2}\sigma(\mathcal X_+)\leq \sigma(\mathcal Z_+)
\leq{\nu_{\max}+\varepsilon\over 2}\sigma(\mathcal X_+)\quad\text{near }\WFh(\Pi\mathbf u(t)),
$$
and both inequalities are strict on $W'\cap\WFh(\Pi\mathbf u(t))$.
By\switch{~\cite[Lemma~8.7]{nhp}}{ Lemma~\ref{l:image-garding}}, we deduce that
$$
\begin{gathered}
\langle \mathcal Z_+\Pi\mathbf u(t),\Pi\mathbf u(t)\rangle\geq {\nu_{\min}-\varepsilon\over 2}\langle\mathcal X_+\Pi\mathbf u(t),\Pi\mathbf u(t)\rangle
+\|A_0\Pi\mathbf u(t)\|_{L^2}^2-\mathcal O(h^\infty),\\
\langle\mathcal Z_+\Pi\mathbf u(t),\Pi\mathbf u(t)\rangle\leq {\nu_{\max}+\varepsilon\over 2}\langle\mathcal X_+\Pi\mathbf u(t),\Pi\mathbf u(t)\rangle
-\|A_0\Pi\mathbf u(t)\|_{L^2}^2+\mathcal O(h^\infty) 
\end{gathered}
$$
By~\eqref{e:bear}, we find
$$
\begin{gathered}
(\partial_t +(\nu_{\min}-\varepsilon))\langle\mathcal X_+\Pi\mathbf u(t),\Pi\mathbf u(t)\rangle
\leq\mathcal O(h^\infty),\\
(\partial_t +(\nu_{\max}+\varepsilon))\langle\mathcal X_+\Pi\mathbf u(t),\Pi\mathbf u(t)\rangle
\geq-\mathcal O(h^\infty).
\end{gathered}
$$
Therefore,
$$
\begin{gathered}
\langle\mathcal X_+\Pi\mathbf u(t),\Pi\mathbf u(t)\rangle\leq Ce^{-(\nu_{\min}-\varepsilon)t}
\langle\mathcal X_+\Pi\mathbf u(t_1),\Pi\mathbf u(t_1)\rangle+\mathcal O(h^\infty),\\
\langle\mathcal X_+\Pi\mathbf u(t),\Pi\mathbf u(t)\rangle\geq C^{-1}e^{-(\nu_{\max}+\varepsilon)t}
\langle\mathcal X_+\Pi\mathbf u(t_1),\Pi\mathbf u(t_1)\rangle-\mathcal O(h^\infty).
\end{gathered}
$$
To prove~\eqref{e:inj-1} and~\eqref{e:inj-2},
it remains to note that
$$
\begin{gathered}
\langle\mathcal X_+\Pi\mathbf u(t),\Pi\mathbf u(t)\rangle\geq C^{-1}\|A_0\Pi\mathbf u(t)\|_{L^2}^2-\mathcal O(h^\infty),\\
\langle\mathcal X_+\Pi\mathbf u(t),\Pi\mathbf u(t)\rangle\leq C\|A_0\Pi\mathbf u(t)\|_{L^2}^2+\mathcal O(h^\infty);
\end{gathered}
$$
the first of these statements follows by\switch{~\cite[Lemma~8.7]{nhp}}{ Lemma~\ref{l:image-garding}}
and the second one is arranged by choosing $A_0$ to be elliptic on $\WFh(\mathcal X_+)$.
\end{proof}

\begin{proof}[Proof of Theorem~\ref{t:internal-waves}]
To construct the component $u_\Pi(t)$, we use $\Pi\mathbf u(t)$ together
with the Schr\"odinger propagator $e^{-itP/h}$. Since $P^*=P$ and $P$ is compactly
supported and compactly microlocalized, the operator $e^{-itP/h}$ quantizes
the flow $e^{tH_p}$ in the sense of\switch{~\cite[Proposition~3.1]{nhp}}{ Proposition~\ref{l:schrodinger}}.
Since $\WFh(\Pi\mathbf u(t))\subset\Gamma_+$, we have by~\eqref{e:pier},
\begin{equation}
  \label{e:u-pi-sol0}
(hD_t+P)e^{-it_1P/h}\Pi\mathbf u(t)=\mathcal O(h^\infty)\quad\text{on }X_{\delta_1},\quad
t\geq t_1,
\end{equation}
if $t_1$ is large enough so that
\begin{equation}
  \label{e:panisse}
e^{-t_1H_p}(\Gamma_+\cap X_{\delta_1}\cap \{|p-1|<\delta_1/4\})\subset W';
\end{equation}
such $t_1$ exists by\switch{~\cite[Lemma~4.1]{nhp}}{ Lemma~\ref{l:the-flow}}.
We then take
an elliptic parametrix $\widetilde{\mathcal S}'$ of $\widetilde{\mathcal S}$
near $\widetilde{\mathcal U}$ (see\switch{~\cite[Proposition~3.3]{nhp}}{ Proposition~\ref{l:eparametrix}}) and define 
\begin{equation}
  \label{e:u-Pi}
u_\Pi(t):=\widetilde{\mathcal S}'(e^{-it_1P/h}\Pi\mathbf u(t-t_1)),\quad
t\in [t_0-1,T\log(1/h)],\quad
t_0:=2t_1+1.
\end{equation}
Then by~\eqref{e:fac-2} and~\eqref{e:u-pi-sol0} we get
$$
h^2\Box_{\tilde g}u_\Pi=\mathcal O(h^\infty)\quad\text{on }X_{\delta_1},
$$
uniformly in $t\in [t_0,T\log(1/h)]$.
Put
$$
u_R(t):=u(t)-u_\Pi(t),\quad
t\in [t_0,T\log(1/h)],
$$
then $h^2\Box_{\tilde g}u_R=\mathcal O(h^\infty)$ on $X_{\delta_1}$ as well.

It remains to prove~\eqref{e:ini-1}--\eqref{e:ini-5}. Since $\WFh(\Pi\mathbf u(t))\subset\Gamma_+$
and by~\eqref{e:panisse}, we find
$$
\|\widetilde{\mathcal S}u_\Pi(t)\|_{L^2(X_{\delta_1})}\leq C\|A_0\Pi\mathbf u(t-t_1)\|_{L^2}+\mathcal O(h^\infty);
$$
here $A_0$ is the operator from Proposition~\ref{l:internal-waves-1}.
Since $[P,\Pi]=\mathcal O(h^\infty)$ microlocally near $\widehat W\times\widehat W$, and by~\eqref{e:kentucky-1}
and~\eqref{e:kentucky-2.5}
(replacing $t_1$ by $s\in [0,t_1]$ in the definition of $u_\Pi$ and differentiating in $s$)
we get
\begin{equation}
  \label{e:empire}
\widetilde{\mathcal S}u_\Pi(t)=\Pi\mathbf u(t)+\mathcal O(h^\infty)\quad\text{microlocally near }\widehat W.
\end{equation}
Therefore,
$$
\|A_0\Pi\mathbf u(t)\|_{L^2}\leq C\|\widetilde{\mathcal S}u_\Pi(t)\|_{L^2(X_{\delta_1})}+\mathcal O(h^\infty).
$$
Next, by~\eqref{e:we-outgoing-2} each backwards flow line of $e^{tH_p}$ starting in $X_{2\delta_1}$ either stays forever
in $X_{2\delta_1}$ or reaches the complement of $\WFh(\mathbf u(t))$~-- see the proof of Proposition~\ref{l:kentucky}.
By propagation of singularities for the equation~\eqref{e:kentucky-1}, we find
$$
\|A\mathbf u(t_1)\|_{L^2}\leq C\|\mathbf u(0)\|_{L^2(X_{\delta_1})}+\mathcal O(h^\infty),\quad
A\in\Psi^{\comp}(X_{2\delta_1}).
$$
Also, for $t_1$ large enough, each flow line $\gamma(t)$, $t\in [-t_1,0]$, of $H_p$
such that $\gamma(0)\in X_{\delta_1}$ either satisfies $\gamma(-t_1)\in W'$
and $\gamma([-t_1,0])\subset X_{\delta_1}$, or there exists $s\in [-t_1,0]$ such that
$\gamma(s)\not\in\WFh(\mathbf u(t))$ for $t\in [t_1,T\log(1/h)]$ and $\gamma([s,0])\subset X_{\delta_1}$.
This is true since if $\gamma(s)\in\WFh(\mathbf u(t))$, then $\gamma([s-t_1,s])\subset X_{\delta_1}$,
see the proof of Proposition~\ref{l:kentucky}.
By propagation of singularities for~\eqref{e:kentucky-1}, we get
\begin{equation}
  \label{e:alabama}
\|\mathbf u(t)\|_{L^2(X_{\delta_1})}\leq C\|A_1\mathbf u(t-t_1)\|_{L^2}+\mathcal O(h^\infty),\quad
t\in [2t_1,T\log(1/h)].
\end{equation}
By~\eqref{e:kentucky-1} and~\eqref{e:u-pi-sol0}, we have
$(hD_t+P)(\widetilde {\mathcal S} u_R(t))=\mathcal O(h^\infty)$ on $X_{\delta_1}$.
Using propagation of singularities for this equation in
a manner similar to~\eqref{e:alabama}, we obtain by~\eqref{e:empire}
$$
\|\widetilde{\mathcal S} u_R(t)\|_{L^2(X_{\delta_1})}\leq C\|A_1(1-\Pi)\mathbf u(t-t_1)\|_{L^2}+\mathcal O(h^\infty),\quad
t\in [2t_1,T\log(1/h)].
$$
Combining these estimates with~\eqref{e:inj-1}--\eqref{e:inj-5}, we arrive to
$$
\begin{aligned}
\|\widetilde{\mathcal S}u_\Pi(t_0)\|_{L^2(X_{\delta_1})}&\leq Ch^{-1/2}\|\mathbf u(0)\|_{L^2(X_{\delta_1})}+\mathcal O(h^\infty),\\
\|\widetilde{\mathcal S}u_\Pi(t)\|_{L^2(X_{\delta_1})}&\leq Ce^{-(\nu_{\min}-\varepsilon)t/2}\|\widetilde{\mathcal S}u_\Pi(t_0)\|_{L^2(X_{\delta_1})}
+\mathcal O(h^\infty),\\
\|\widetilde{\mathcal S}u_\Pi(t)\|_{L^2(X_{\delta_1})}&\geq C^{-1}e^{-(\nu_{\max}+\varepsilon)t/2}\|\widetilde{\mathcal S}u_\Pi(t_0)\|_{L^2(X_{\delta_1})}
-\mathcal O(h^\infty),\\
\|\widetilde{\mathcal S}u_R(t)\|_{L^2(X_{\delta_1})}&\leq Ch^{-1}e^{-(\nu_{\min}-\varepsilon)t}\|\mathbf u(0)\|_{L^2(X_{\delta_1})}
+\mathcal O(h^\infty),\\
\|\mathbf u(t)\|_{L^2(X_{\delta_1})}&\leq Ce^{\varepsilon t}\|\mathbf u(0)\|_{L^2(X_{\delta_1})}+\mathcal O(h^\infty),
\end{aligned}
$$
holding uniformly in $t\in [t_0-1,T\log(1/h)]$. To obtain~\eqref{e:ini-1}--\eqref{e:ini-5} from here,
we need to remove the operator $\widetilde{\mathcal S}$ from the estimates; for that,
we can use the fact that $\widetilde S$ is bounded uniformly in $h$ on $L^2_{t,x}$
together with
the equivalency of the norms $h\|\cdot\|_{L^\infty_t\mathcal E_x}$ and
$\|\cdot\|_{L^2_{t,x}}$  for solutions
of the wave equation~\eqref{e:wewe} given by~\eqref{e:equivl2}
and the functions of interest being microlocalized at frequencies $\sim h^{-1}$.
\end{proof}

\section{Applications to Kerr(--de Sitter) metrics}
  \label{s:kds-metrix}

\subsection{General properties}
  \label{s:kdsu-the-metric}

The Kerr(--de Sitter) metric in the Boyer--Lindquist coordinates is given by the formulas~\cite{carter}
$$
\begin{gathered}
g=-\rho^2\Big({dr^2\over\Delta_r}+{d\theta^2\over\Delta_\theta}\Big)
-{\Delta_\theta\sin^2\theta\over (1+\alpha)^2\rho^2}
(a\,dt-(r^2+a^2)d\varphi)^2\\
+{\Delta_r\over (1+\alpha)^2\rho^2}(dt-a\sin^2\theta\,d\varphi)^2.
\end{gathered}
$$
Here $M>0$ denotes the mass of the black hole, $a$ its angular speed of rotation,
and $\Lambda\geq 0$ is the cosmological constant (with $\Lambda=0$ in
the Kerr case and $\Lambda>0$ in the Kerr--de Sitter case);
$$
\begin{gathered}
\Delta_r=(r^2+a^2)\Big(1-{\Lambda r^2\over 3}\Big)-2Mr,\quad
\Delta_\theta=1+\alpha\cos^2\theta,\\
\rho^2=r^2+a^2\cos^2\theta,\quad
\alpha={\Lambda a^2\over 3}.
\end{gathered}
$$
The metric is originally defined on
$$
\widetilde X_0:=\mathbb R_t\times X_0,\quad
X_0:=(r_-,r_+)_r\times \mathbb S^2,
$$
here $\theta\in [0,\pi]$ and $\varphi\in \mathbb S^1=\mathbb R/(2\pi\mathbb Z)$
are the spherical coordinates on $\mathbb S^2$. The numbers $r_-<r_+$ are the roots
of $\Delta_r$ defined below; in particular, $\Delta_r>0$ on $(r_-,r_+)$
and $\pm \partial_r\Delta_r(r_\pm)<0$. The metric becomes singular on the surfaces
$\{r=r_\pm\}$, known as the \emph{event horizons}; however, this can be fixed
by a change of coordinates, see~\S\ref{s:kdsu-waves}.

The Kerr(--de Sitter) family
admits the scaling $M\mapsto sM,\Lambda\mapsto s^{-2}\Lambda,a\mapsto sa,r\mapsto sr,t\mapsto st$ for $s>0$;
therefore, we often consider the parameters $a/M$ and $\Lambda M^2$ invariant under this scaling.
We assume that $a/M,\Lambda M^2$ lie in a neighborhood
of the Schwarzschild(--de Sitter) case~\eqref{e:sds-bound} or the Kerr case~\eqref{e:kerr-bound}.
Then for $\Lambda>0$, $\Delta_r$ is a degree 4 polynomial
with real roots $r_1<r_2<r_-<r_+$, with $r_->M$. For $\Lambda=0$, $\Delta_r$ is a degree 2 polynomial with real roots
$r_1<M<r_-$; we put $r_+=\infty$. The general set of $\Lambda$ and $a$ for which $\Delta_r$ has all real roots as above
was studied numerically in~\cite[\S 3]{ak-ma}, and is pictured on Figure~\ref{f:kdsu-gaps}(a) in the introduction.
Note that in~\cite{ak-ma}, the roots are labeled
$r_{--}<r_-<r_+<r_C$; we do not adopt this (perhaps more standard) convention in favor of the notation
of\switch{~\cite{skds,zeeman,vasy1},}{ the previous chapters and~\cite{vasy1},}
and since the roots $r_1,r_2$ are irrelevant in our analysis.

The symbol $\tilde p$ defined in~\eqref{e:tilde-p} using the dual metric is
(denoting by $\tau$ the momentum corresponding to $t$)
$$
\begin{gathered}
\tilde p=\rho^{-2}G,\quad
G=G_r+G_\theta,\\
G_r=\Delta_r\xi_r^2-{(1+\alpha)^2\over\Delta_r}((r^2+a^2)\tau+a\xi_\varphi)^2,\\
G_\theta=\Delta_\theta\xi_\theta^2+{(1+\alpha)^2\over\Delta_\theta\sin^2\theta}(a\sin^2\theta\,\tau+\xi_\varphi)^2.
\end{gathered}
$$
Note that
\begin{equation}
  \label{e:cyclic}
\partial_{(t,\varphi,\theta,\xi_\theta)}G_r=0,\quad
\partial_{(t,\varphi,r,\xi_r)}G_\theta=0,
\end{equation}
therefore $\{G_r,G_\theta\}=0$ and $G_\theta,\tau,\xi_\varphi$ are conserved quantities
for the geodesic flow~\eqref{e:rescaled-flow}.

To handle the poles $\{\theta=0\}$ and $\{\theta=\pi\}$, where the spherical coordinates $(\theta,\varphi)$
break down, introduce new coordinates (in a neighborhood of either of the poles)
\begin{equation}
  \label{e:poles-coordinates}
x_1=\sin\theta\cos\varphi,\quad
x_2=\sin\theta\sin\varphi;
\end{equation}
note that $\sin^2\theta=x_1^2+x_2^2$ is a smooth function in this coordinate system.
For the corresponding momenta $\xi_1,\xi_2$, we have
$$
\xi_\theta=(x_1\xi_1+x_2\xi_2)\cot\theta,\quad
\xi_\varphi=x_1\xi_2-x_2\xi_1,
$$
note that $\xi_\varphi$ is a smooth function vanishing at the poles.
Then $G_r, G_\theta$ are smooth functions near the poles, with
\begin{equation}
  \label{e:g-theta-poles}
G_\theta=(1+\alpha)(\xi_1^2+\xi_2^2)\quad\text{when }x_1=x_2=0.
\end{equation}
The vector field $\partial_t$ is not timelike inside the ergoregion, described by the inequality
\begin{equation}
  \label{e:ergogre}
\Delta_r\leq a^2\Delta_\theta\sin^2\theta.
\end{equation}
For $a\neq 0$, this region is always nonempty. However, the
covector $dt$ is always timelike:
\begin{equation}
  \label{e:kds-spacelike}
G|_{\xi=dt}=(1+\alpha)^2\Big({a^2\sin^2\theta\over\Delta_\theta}-{(r^2+a^2)^2\over\Delta_r}\Big)<0,
\end{equation}
since $\Delta_r<r^2+a^2$.

We now verify the geometric assumptions~\eqref{w:basic}--\eqref{w:convex} of~\S\ref{s:kdsu-assumptions}.
Assumptions~\eqref{w:basic}--\eqref{w:spacelike} have been established already;
assumption~\eqref{w:convex} is proved by
\begin{prop}
  \label{l:mu-conv}
Consider the function $\mu(r)\in C^\infty(r_-,r_+)$ defined by
\begin{equation}
  \label{e:mu-defined}
\mu(r):={\Delta_r(r)\over r^4}.
\end{equation}
Then there exists $\delta_0>0$ such that
for each $(\tilde x,\tilde\xi)\in T^*\widetilde X_0$,
\begin{equation}
  \label{e:mu-conv}
\mu(\tilde x)<\delta_0,\
\tilde\xi\neq 0,\
\tilde p(\tilde x,\tilde\xi)=0,\
H_{\tilde p}\mu(\tilde x,\tilde\xi)=0\ \Longrightarrow\
H_{\tilde p}^2\mu(\tilde x,\tilde\xi)<0.
\end{equation}
Moreover, $\delta_0$ can be chosen to depend continuously on $M,\Lambda,a$.
\end{prop}
\begin{proof}
First of all, we calculate
\begin{equation}
  \label{e:dallas}
\partial_r\mu(r)=-{4\Delta_r-r\partial_r\Delta_r\over r^5},\quad
4\Delta_r-r\partial_r\Delta_r=2((1-\alpha)r^2-3Mr+2a^2),
\end{equation}
therefore $\partial_r\mu(r)<0$ for $\alpha\leq 1/2$ and $r>6M$.
Since $\partial_r\Delta_r(r_\pm)\neq 0$, we see that
for $\delta_0$ small enough and $\mu(r)<\delta_0$,
we have $\partial_r\mu(r)\neq 0$. Therefore, we can replace
the condition $H_{\tilde p}\mu=0$ in~\eqref{e:mu-conv} by $H_{\tilde p}r=0$,
which implies that $\xi_r=0$; in this case,
$H_{\tilde p}^2\mu$ has the same sign as $-\partial_r\mu\partial_r G_r$. We
calculate for $\xi_r=0$,
\begin{equation}
  \label{e:Psi-r}
\begin{gathered}
\partial_r G_r=-{(1+\alpha)^2((r^2+a^2)\tau+a\xi_\varphi)\over\Delta_r^2}\Psi(r),\\
\Psi(r):=4r\tau\Delta_r-((r^2+a^2)\tau+a\xi_\varphi)\partial_r\Delta_r.
\end{gathered}
\end{equation}
Next, denote
\begin{equation}
  \label{e:AB}
A := (r^2+a^2)\tau+a\xi_\varphi,\quad
B := a\sin^2\theta \,\tau+\xi_\varphi,
\end{equation}
then
\begin{equation}
  \label{e:Psi-r'}
\rho^2\tau=A-aB,\quad
\Psi={(4r\Delta_r-\rho^2\partial_r\Delta_r)A-4ar\Delta_r B\over\rho^2}.
\end{equation}
Using the equation $\tilde p=0$, we get
\begin{equation}
  \label{e:g=0.1}
{A^2\over\Delta_r}\geq {B^2\over\Delta_\theta\sin^2\theta}
\quad\text{on }\{\tilde p=\xi_r=0\}\cap \{0<\theta<\pi\}.
\end{equation}
Since $\Delta_\theta\sin^2\theta\leq 1$ everywhere for $\alpha\leq 1$, and $B=0$ for $\sin\theta=0$, we find
\begin{equation}
  \label{e:AB-eq}
A^2\geq\Delta_r B^2.
\end{equation}
In particular, we see that $A\neq 0$, since otherwise $B=0$, implying that $\tau=\xi_\varphi=0$
and thus $\tilde \xi=0$ since $\tilde p=\xi_r=0$. Now, $H_{\tilde p}^2\mu$ has the same sign as
\begin{equation}
  \label{e:nebraska}
\partial_r\mu((4r\Delta_r-\rho^2\partial_r\Delta_r)A^2-4ar\Delta_r AB).
\end{equation}
We now calculate by~\eqref{e:dallas} and since $\partial_r\Delta_r\leq 2r$, for $\alpha\leq 1/2$
$$
\begin{gathered}
4r\Delta_r-\rho^2\partial_r\Delta_r=2r((1-\alpha)r^2-3Mr+2a^2)-a^2\cos^2\theta\partial_r\Delta_r\\
\geq r(r^2-6Mr+2a^2),
\end{gathered}
$$
and thus, since $\Delta_r\leq r^2+a^2$ and $|a|<M$, and by~\eqref{e:AB-eq},
$$
\begin{gathered}
(4r\Delta_r-\rho^2\partial_r\Delta_r)A^2-4ar\Delta_r AB\geq A^2 r(r^2-6Mr+2a^2-4|a|\sqrt{\Delta_r})\\
\geq A^2r (r^2-10Mr-4M^2).
\end{gathered}
$$
We see that~\eqref{e:mu-conv} holds for $r$ large enough, namely $r>14M$.

We now assume that $r\leq 14M$ and $\mu<\delta_0$. Then (here the constants
do not depend on $\delta_0$ and are locally uniform in
$M,\Lambda,a$)
$$
\Delta_r=\mathcal O(\delta_0),\quad
|\partial_r\Delta_r|\geq C^{-1},\quad
\partial_r\mu=r^{-4}\partial_r\Delta_r+\mathcal O(\delta_0).
$$
Then for $\delta_0$ small enough, by~\eqref{e:AB-eq} the expression \eqref{e:nebraska} has the same sign as
$$
A^2\partial_r\Delta_r(-\rho^2\partial_r\Delta_r+\mathcal O(\sqrt{\delta_0}))<0,
$$
as required.
\end{proof}

\subsection{Structure of the trapped set}
  \label{s:kdsu-trapping}

We now study the structure of trapping for Kerr(--de Sitter) metrics, summarized in the following
\begin{prop}
  \label{l:kds-trapping}
For $(\Lambda M^2,a/M)$ in a neighborhood of the union of~\eqref{e:sds-bound} and~\eqref{e:kerr-bound},
assumptions~\eqref{w:positive}--\eqref{w:nhp} of~\S\ref{s:kdsu-assumptions} are satisfied,
with $\mu_{\max}=0$ (see~\eqref{e:tilde-mu-max}) and the trapped set (see Definition~\ref{d:time-trapping}) given by
\begin{equation}
  \label{e:kds-trapped-set}
\widetilde K=\{G=\xi_r=\partial_r G_r=0,\ \tilde\xi\neq 0\}\subset T^*\widetilde X_0\setminus 0.
\end{equation}
\end{prop}

\noindent\textbf{Remark}. The assumptions on $M,\Lambda,a$ can quite possibly be relaxed.
The only parts of the proof that need us to be in a neighborhood of~\eqref{e:sds-bound}
or~\eqref{e:kerr-bound} are~\eqref{e:houston} and~\eqref{e:kdstr-1}.
Several other statements require that $\alpha$ is small (in particular, \eqref{e:kdstr-4}
requires $\alpha<{\sqrt 2 - 1 \over \sqrt 2 + 1}$), but this is true for
the full admissible range of parameters depicted on Figure~\ref{f:kdsu-gaps}(a) in the introduction.
However, if $\Lambda,a$ are large enough so that
\begin{equation}
  \label{e:weird-region}
r\in (r_-,r_+),\ \partial_r\Delta_r(r)=0\ \Longrightarrow\ 
\Delta_r(r)=a^2,
\end{equation}
then the trapped set contains points
with $\tau=0$ (and also $\xi_r=\xi_\theta=0$, $\theta=\pi/2$, $\xi_\varphi\neq 0$),
which prevent us from having a meromorphic continuation
of the resolvent and violate the required assumption~\eqref{w:positive} in~\S\ref{s:kdsu-assumptions}~-- see the discussion
preceding~\cite[(6.13)]{vasy1}. The set of values of $\Lambda,a$ satisfying~\eqref{e:weird-region}
is pictured as the dashed line on Figure~\ref{f:kdsu-gaps}(a).

\noindent\textbf{Remark}.
Some parts of Proposition~\ref{l:kds-trapping} have previously been verified in~\cite[\S6.4]{vasy1}
in the case $|a|<{\sqrt 3\over 2}M$ and under the additional assumption~\cite[(6.13)]{vasy1}.

We start by analysing the set $\widetilde K$ defined by~\eqref{e:kds-trapped-set};
the fact that $\widetilde K$ is indeed the trapped set is established later,
in Proposition~\ref{l:true-trapping}.
We first note that $\widetilde K$ is a closed conic subset of $\{\tilde p=0\}\setminus 0$,
invariant under the flow~\eqref{e:rescaled-flow}; indeed, $\xi_r=0$ implies $H_{\tilde p}r=0$,
$\partial_r G_r=0$ implies $H_{\tilde p}\xi_r=0$, $H_{\tilde p}\tau=H_{\tilde p}\xi_\varphi=0$
everywhere, and $\partial_r G_r$ depends only on $r,\xi_r,\tau,\xi_\varphi$.

By~\eqref{e:Psi-r}, and since $(r^2+a^2)\tau+a\xi_\varphi=\tilde p=0$ implies $\tilde\xi=0$, we see
that
\begin{equation}
  \label{e:psi=0}
\Psi=0\quad\text{on }\widetilde K.
\end{equation}
Assumption~\eqref{w:positive} in~\S\ref{s:kdsu-assumptions} follows from the inequality
\begin{equation}
  \label{e:houston}
\tau((r^2+a^2)\tau+a\xi_\varphi)>0\quad\text{on }\widetilde K.
\end{equation}
For the Schwarzschild(--de Sitter) case~\eqref{e:sds-bound}, this is trivial
(noting that $\tau=0$ implies $\tilde\xi=0$); for the Kerr case~\eqref{e:kerr-bound},
it follows from~\eqref{e:psi=0} together with the fact that $\partial_r\Delta_r>0$.
The general case now follows by perturbation, using that, by Proposition~\ref{l:mu-conv},
$\widetilde K$ is contained in a fixed compact subset of $X_0$.

We next claim that
\begin{equation}
  \label{e:kdstr-1}
\partial_r^2 G<0\quad\text{on }\widetilde K.
\end{equation}
By~\eqref{e:Psi-r}, this is equivalent to requiring that $\tau\partial_r\Psi>0$ on $\widetilde K$. Now,
in either of the cases~\eqref{e:sds-bound} or~\eqref{e:kerr-bound}, we calculate
\begin{equation}
  \label{e:PsiPsi}
\Psi(r)=2(\tau r^3-3M\tau r^2
+a(a\tau-\xi_\varphi)r
+Ma(a\tau+\xi_\varphi)).
\end{equation}
In particular,
$$
\Psi(M)=4M\tau(a^2-M^2),\quad
\partial_r^2\Psi(r)=12\tau(r-M).
$$
Since $|a|<M$, we see that
$$
\tau\Psi(M)<0;\quad
\tau\partial_r^2\Psi(r)>0\text{ for }r>M.
$$
Therefore, if $r>r_->M$ and $\Psi(r)=0$, then $\tau\partial_r\Psi(r)>0$
and we get~\eqref{e:kdstr-1} in the cases~\eqref{e:sds-bound} and~\eqref{e:kerr-bound};
the general case follows by perturbation, similarly to~\eqref{e:houston}.

To study the behavior of $\widetilde K$ in the angular variables, we introduce the equatorial set
\begin{equation}
  \label{e:tilde-k-e}
\widetilde K_e:=\widetilde K\cap \{\theta=\pi/2,\ \xi_\theta=0\}.
\end{equation}
This is a closed conic subset of $\widetilde K$ invariant under the flow~\eqref{e:rescaled-flow}
(which is proved similarly to the invariance of $\widetilde K$).
We have
\begin{equation}
\label{e:kdstr-2}
\partial_{\xi_\varphi}G\neq 0\quad\text{on }\widetilde K_e.
\end{equation}
Indeed,
\begin{equation}
  \label{e:pxvG}
\partial_{\xi_\varphi} G=2(1+\alpha)^2\Big(
-{a((r^2+a^2)\tau+a\xi_\varphi)\over\Delta_r}+a\tau+\xi_\varphi\Big)\quad\text{on }\{\theta=\pi/2\}.
\end{equation}
Also, the equation $G=0$ implies
\begin{equation}
  \label{e:g=0.2}
{((r^2+a^2)\tau+a\xi_\varphi)^2\over\Delta_r}= {(a\sin^2\theta\,\tau+\xi_\varphi)^2\over\Delta_\theta\sin^2\theta}
\quad\text{on }\widetilde K\cap \{\xi_\theta=0\}.
\end{equation}
Putting $\theta=\pi/2$ into~\eqref{e:g=0.2}, we solve for $\Delta_r$ and substitute it into~\eqref{e:pxvG}, obtaining
\begin{equation}
  \label{e:gould}
\partial_{\xi_\varphi} G=2(1+\alpha)^2{r^2\tau(a\tau+\xi_\varphi)\over (r^2+a^2)\tau+a\xi_\varphi}\neq 0\quad\text{on }\widetilde K_e,
\end{equation}
implying~\eqref{e:kdstr-2}.

At the poles $\{\theta=0,\pi\}$, we have
\begin{equation}
  \label{e:kdstr-3}
|\partial_{\xi_1}G|+|\partial_{\xi_2}G|>0.
\end{equation}
This follows immediately from~\eqref{e:g-theta-poles},
as $\xi_1=\xi_2=0$ would imply $G_\theta=0$, which is impossible given that $\xi_r=0$, $G=0$,
and $\tilde\xi\neq 0$.

Finally, we claim that
\begin{gather}
  \label{e:kdstr-4}
\widetilde K\cap \{\xi_\theta=\partial_{\theta}G=0\}\cap \{0<\theta<\pi\}=\widetilde K_e,\\
  \label{e:kdstr-5}
\partial_\theta^2 G>0\quad\text{on }\widetilde K_e.
\end{gather}
To see this, note that  $0<\Delta_r<r^2+a^2$,
$\Delta_\theta\geq 1$, and
$(r^2+a^2)\tau+a\xi_\varphi\neq 0$ by~\eqref{e:houston}; we get from~\eqref{e:g=0.2}
$$
((r^2+a^2)\tau+a\xi_\varphi)^2<(r^2+a^2){(a\sin^2\theta\,\tau+\xi_\varphi)^2\over\sin^2\theta}
\quad\text{on }\widetilde K\cap \{\xi_\theta=0\},
$$
or, using that $|a|<M<r$,
\begin{equation}
  \label{e:xi-phi-bound}
{\xi_{\varphi}^2\over\sin^2\theta}>(r^2+a^2)\tau^2>2a^2\tau^2
\quad\text{on }\widetilde K\cap \{\xi_\theta=0\}.
\end{equation}
Next, if $\xi_\theta=0$, then
$$
\partial_\theta G={2(1+\alpha)^2(a\sin^2\theta\tau+\xi_\varphi)\cos\theta\over \Delta_\theta^2\sin^3\theta}
((1+\alpha)a\sin^2\theta\,\tau-(1+\alpha\cos(2\theta))\xi_\varphi).
$$
In particular, using~\eqref{e:xi-phi-bound} we obtain~\eqref{e:kdstr-5}
for $\alpha=0$:
$$
\partial_\theta^2 G=2(\xi_\varphi^2-a^2\tau^2)>0\quad\text{on }\widetilde K_e,
$$
and the case of small $\alpha$ follows by perturbation.
It remains to prove~\eqref{e:kdstr-4}. Assume the contrary, that
$\partial_\theta G=0$, $\xi_\theta=0$, but $\theta\neq\pi/2$. By~\eqref{e:g=0.2},
$a\sin^2\theta\tau+\xi_\varphi\neq 0$; therefore,
$(1+\alpha)a\sin^2\theta\,\tau=(1+\alpha\cos(2\theta))\xi_\varphi$.
Combining this with~\eqref{e:xi-phi-bound}, we get
$(1+\alpha)\sin\theta>\sqrt{2}(1+\alpha\cos(2\theta))$, which implies
that $(1+\alpha)>\sqrt{2}(1-\alpha)$, a contradiction with the fact that $\alpha$ is small.

It follows from~\eqref{e:kdstr-1}, \eqref{e:kdstr-2}, \eqref{e:kdstr-3},
and~\eqref{e:kdstr-4} that at each point
of $\widetilde K$ the matrix of partial derivatives
$G,\xi_r,\partial_r G$ 
in the variables $(r,\xi_r,*)$, where $*$ stands
for one of $\theta,\xi_\theta,\xi_\varphi,\xi_1,\xi_2$,
is invertible. This gives
\begin{prop}
  \label{l:kds-trapping-smooth}
The set $\widetilde K$ defined by~\eqref{e:kds-trapped-set} is a smooth codimension 2 submanifold
of $\{\tilde p=0\}\setminus 0$, and its projection $\widehat K$ onto the
$\hat x=(t,\theta,\varphi),\hat\xi=(\tau,\xi_\theta,\xi_\varphi)$
variables is a smooth codimension 1 submanifold of $T^*(\mathbb R\times \mathbb S^2)$.
\end{prop}
We now study the global dynamics of the flow, relating it to the set $\widetilde K$.
Take $(\tilde x^0,\tilde \xi^0)\in \{\tilde p=0\}\setminus 0$ and
let $(\tilde x(t),\tilde\xi(t))$ be the corresponding Hamiltonian trajectory
of~\eqref{e:rescaled-flow}. Consider the function
$$
\Phi^0(r)=G_r(\tilde x^0,\tilde \xi^0)+(1+\alpha)^2{((r^2+a^2)\tau^0+a\xi^0_\varphi)^2\over\Delta_r(r)}.
$$
Note that $G_r(\tilde x(t),\tilde\xi(t)),\tau(t),\xi_\varphi(t)$ are constant in $t$
and $(r(t),\xi_r(t))$ is a rescaled Hamiltonian flow trajectory of
$$
H^0(r,\xi_r):=\Delta_r(r)\xi_r^2-\Phi^0(r);
$$
in particular, $(r(t),\xi_r(t))$ solve the equation
\begin{equation}
  \label{e:phi-eqq}
\Delta_r(r)\xi_r^2=\Phi^0(r).
\end{equation}
The key property of $\Phi^0$ is given by
\begin{prop}
For each $r\in (r_-,r_+)$,
\begin{equation}
  \label{e:Phi-helper-1}
\Phi^0(r)\geq 0,\
\partial_r\Phi^0(r)=0\
\Longrightarrow\
\partial_r^2\Phi^0(r)>0.
\end{equation}
\end{prop}
\begin{proof}
Assume that $\Phi^0(r)\geq 0$. Then
we can find $(\tilde x^1,\tilde\xi^1)\in T^*\widetilde X^0$ such that
$(t^1,\theta^1,\varphi^1)=(t^0,\theta^0,\varphi^0)$,
$\tilde r^1=r$, $\tau^1=\tau^0$, $\xi_\varphi^1=\xi_\varphi^0$,
$\xi_r^1=0$, and $\tilde p(\tilde x^1,\tilde\xi^1)=0$; indeed,
it suffices to start with $(\tilde x^0,\tilde\xi^0)$, put $r^1=r,\xi^1_r=0$, and
change the $\xi^1_\theta$ component (or one of $\xi^1_1,\xi^1_2$ components
if we are at the poles of the sphere) so that $G_\theta(\tilde x^1,\tilde\xi^1)
=G_\theta(\tilde x^0,\tilde\xi^0)+\Phi^0(r)$.
If additionally $\partial_r\Phi^0(r)=0$, then $(\tilde x^1,\tilde\xi^1)\in\widetilde K$;
it remains to apply~\eqref{e:kdstr-1}.
\end{proof}
We now consider the following two cases:

\noindent\textbf{Case 1}: $|\Phi^0(r)|+|\partial_r\Phi^0(r)|>0$ for all $r\in (r_-,r_+)$.
In this case, the set of solutions to~\eqref{e:phi-eqq} is a closed one-dimensional submanifold
of $T^*(r_-,r_+)$ and the Hamiltonian field of $H^0$ is nonvanishing on this manifold.
This manifold has no compact connected components, as the function $\Phi^0(r)$ cannot
achieve a local maximum on it by~\eqref{e:Phi-helper-1}.
It follows that the geodesic $(\tilde x(t),\tilde\xi(t))$ escapes in both time
directions.

\noindent\textbf{Case 2}: there exists $r'\in (r_-,r_+)$ such that
$\Phi^0(r')=\partial_r\Phi^0(r')=0$. Then
$$
(t^0,r',\theta^0,\varphi^0,\tau^0,0,\xi_\theta^0,\xi_\varphi^0)\in\widetilde K,
$$
therefore the projection $(\hat x^0,\hat \xi^0)$ lies in $\widehat K$
(see Proposition~\ref{l:kds-trapping-smooth}). By~\eqref{e:Phi-helper-1},
we see that $\partial_r^2\Phi^0(r')>0$ and $(r-r')\partial_r\Phi^0(r)>0$ for $r\neq r'$.
Then the set of solutions to the equation~\eqref{e:phi-eqq} is equal to the union 
$\Gamma^0_+\cup\Gamma^0_-$, where
$$
\Gamma^0_\pm = \{ \xi_r=\mp \sgn(\tau^0)\sgn(r-r')\sqrt{\Phi^0(r)/\Delta_r(r)}\},
$$
note that $\Gamma^0_\pm$ are smooth one-dimensional submanifolds of $T^*(r_-,r_+)$
intersecting transversely at $(r',0)$. The trajectory
$(\tilde x(t),\tilde\xi(t))$ is trapped as $t\to \mp\infty$
if and only if $(r^0,\xi_r^0)\in\Gamma^0_\pm$. Note that by~\eqref{e:houston},
$\tau^0$ is negative on $\mathcal C_+$ and positive on $\mathcal C_-$.

The analysis of the two cases above implies
\begin{prop}
  \label{l:true-trapping}
The incoming/outgoing tails $\widetilde\Gamma_\pm$ (see Definition~\ref{d:time-trapping})
are given by (here $\widehat K$ is defined in Proposition~\ref{l:kds-trapping-smooth})
$$
\widetilde\Gamma_\pm:=\{(r,\hat x,\xi_r,\hat\xi)\mid
(\hat x,\hat\xi)\in \widehat K,\
\xi_r=\mp\sgn(\hat\tau)\sgn(r-r'_{\hat x,\hat\xi})\sqrt{\Phi_{\hat x,\hat\xi}(r)/\Delta_r(r)}
\},
$$
where
$$
\Phi_{\hat x,\hat\xi}(r)=-G_\theta(\hat x,\hat\xi)+(1+\alpha)^2{((r^2+a^2)\hat\tau+a\hat\xi_\varphi)^2\over\Delta_r(r)},
$$
and $r'_{\hat x,\hat\xi}$ is the only solution to the equation $\Phi_{\hat x,\hat\xi}(r)=0$;
moreover, $\partial_{r}\Phi_{\hat x,\hat\xi}(r'_{\hat x,\hat\xi})=0$
and $\partial_r^2\Phi_{\hat x,\hat\xi}(r'_{\hat x,\hat\xi})>0$.
Furthermore, $\widetilde\Gamma_\pm$ are conic smooth codimension 1 submanifolds of $\{\tilde p=0\}\setminus 0$
intersecting transversely, and their intersection is equal to the set
$\widetilde K$ defined in~\eqref{e:kds-trapped-set}.
\end{prop}
We also see from~\eqref{e:kdstr-1} and the
fact that $\partial_\tau \tilde p\neq 0$ on $\{\tilde p=0\}\setminus 0$
(as follows from~\eqref{e:kds-spacelike}) that the matrix of Poisson brackets
of functions $G,\partial_rG, \xi_r,t$ on $\widetilde K$ is nondegenerate,
which implies that the intersections $\widetilde K\cap \{t=\const\}$ are symplectic
submanifolds of $T^*\widetilde X_0$. Together with Proposition~\ref{l:true-trapping},
this verifies assumptions~\eqref{w:cr} and~\eqref{w:symplectic} of~\S\ref{s:kdsu-assumptions}.

It remains to verify $r$-normal hyperbolicity of the flow $\tilde\varphi^s$
defined in~\eqref{e:rescaled-flow}. We start by showing that
the maximal expansion rate in the directions of the trapped set $\mu_{\max}$,
defined in~\eqref{e:tilde-mu-max}, is equal to zero:
\begin{prop}
  \label{l:angular-flow}
For each $\varepsilon>0$, there exists a constant $C$ such that
for each $v\in T\widetilde K$, 
$$
|d\tilde\varphi^s v|\leq Ce^{\varepsilon |s|}|v|.
$$
Here $|\cdot|$ denotes any fixed smooth homogeneous norm on the fibers of $T\widetilde K$.
\end{prop}
\begin{proof}
Using the group property of the flow,
it suffices to show that for each $\varepsilon>0$
there exists $T>0$ such that for each $v\in T\widetilde K$,
\begin{equation}
  \label{e:hood}
|d\tilde\varphi^T v|< e^{\varepsilon T}|v|.
\end{equation}
Since $\widetilde K$ is a closed conic set, and
$\widetilde K\cap\{\tau=1\}\cap \{t=0\}$ is compact, 
it suffices to show that for each flow line $\gamma(s)$ of~\eqref{e:rescaled-flow} on $\widetilde K$,
there exists $T$ such that~\eqref{e:hood} holds for each $v=v(0)$ tangent to $\widetilde K$ at $\gamma(0)$.
Denote $v(s)=d\varphi^s v(0)$.

If $\gamma(s)$ is a trajectory of~\eqref{e:rescaled-flow} on $T\widetilde K$, then
$r,\xi_r=0,\tau$ are constant on $\gamma(s)$ and the generator of the flow
does not depend on the variable $t$; therefore, it suffices to show~\eqref{e:hood}
for the restriction of the matrix of $d\tilde\varphi^T$ to the
$\partial_\theta,\partial_\varphi,\partial_{\xi_\theta},\partial_{\xi_\varphi}$ variables.
This is equivalent to considering the Hamiltonian flow of $G$ in the $\theta,\varphi,\xi_\theta,\xi_\varphi$
variables only, on $T^*\mathbb S^2$. Recall that the equatorial set $\widetilde K_e=\widetilde K\cap \{\theta=\pi/2,\ \xi_\theta=0\}$ defined in~\eqref{e:tilde-k-e}
is invariant under $\tilde\varphi^s$. We then consider two cases:

\noindent\textbf{Case 1}: $\gamma(s)\not\in \widetilde K_e$ for all $s$. Then
the differentials of $G$ and $\xi_\varphi$ are linearly independent by~\eqref{e:kdstr-3}
and~\eqref{e:kdstr-4}. Since $\{G,\xi_\varphi\}=0$, by Arnold--Liouville
theorem (see for example\switch{~\cite[Proposition~2.8]{nhp}}{ Proposition~\ref{l:arnold-liouville-1}}),
there is a local symplectomorphism from a neighborhood
of $\gamma(s)$ in $T^*\mathbb S^2$ to $T^*\mathbb T^2$, where $\mathbb T^2$ is the two-dimensional torus,
which conjugates $G$ to some function $f(\eta_1,\eta_2)$; here $(y_1,y_2,\eta_1,\eta_2)$
are the canonical coordinates on $T^*\mathbb T^2$. The corresponding evolution of tangent vectors is
given by $\partial_s v_y(s)=\nabla^2 f(\eta(s))v_\eta(s)$, $\partial_s v_\eta(s)=0$,
and~\eqref{e:hood} follows.

\noindent\textbf{Case 2}: $\gamma(s)\in \widetilde K_e$ for all $s$. 
Since $\partial_s v_{\xi_\varphi}(s)=0$
and $\partial_s v$ does not depend on $v_{\varphi(s)}$, it suffices to estimate $v_{\theta}(s),v_{\xi_\theta}(s)$.
We then find
$$
\partial_s v_{\theta}(s)=2v_{\xi_\theta}(s),\quad
\partial_s v_{\xi_\theta}(s)=-\partial_\theta^2 G(\gamma(s))v_{\theta}(s)-\partial^2_{\theta\xi_\varphi}G(\gamma(s))v_{\xi_\varphi}(s).
$$
Now, by~\eqref{e:kdstr-5}, $\partial_\theta^2 G(\gamma(s))$ is a positive constant;
\eqref{e:hood} follows.
\end{proof}
We finally show that the minimal expansion rate $\nu_{\min}$, defined in~\eqref{e:tilde-nu-min},
is positive. By Proposition~\ref{l:true-trapping}, $(\tilde x,\tilde\xi)\in\widetilde\Gamma_\pm$ if
and only if
$$
(\hat x,\hat\xi)\in\widehat K,\quad
\tilde\varphi_\pm(\tilde x,\tilde\xi)=0,
$$
where
$$
\tilde\varphi_\pm(\tilde x,\tilde\xi)=\xi_r\mp \sgn(\partial_\tau G)\sgn(r-r'_{\hat x,\hat\xi})
\sqrt{\Phi_{\hat x,\hat\xi}(r)/\Delta_r(r)}.
$$
Since $H_G$ is tangent to $\widetilde\Gamma_\pm$, we have
$H_G\tilde\varphi_\pm=0$ on $\widetilde\Gamma_\pm$; it follows
that
\begin{equation}
  \label{e:reno}
{H_G\tilde\varphi_\pm(\tilde x,\tilde\xi)\over\partial_\tau G}=\mp\tilde\nu_\pm(\tilde x,\tilde\xi) \tilde\varphi_\pm(\tilde x,\tilde\xi)\quad\text{when }(\hat x,\hat\xi)\in\widehat K,
\end{equation}
for some functions $\tilde\nu_\pm$.
By calculating $\partial_{\xi_r}H_G\tilde\varphi_\pm|_{\widetilde K}$, we find
$\tilde\nu_+|_{\widetilde K}=\tilde\nu_-|_{\widetilde K}=\tilde\nu$, where
\begin{equation}
  \label{e:tilde-nu}
\tilde\nu={\sqrt{-2\Delta_r\partial_r^2G_r}\over |\partial_\tau G|};
\end{equation}
note that $\partial_r^2G_r<0$ on $\widetilde K$ by~\eqref{e:kdstr-1}
and $\partial_\tau G\neq 0$ on $\{\tilde p=0\}\setminus 0$
by assumption~\eqref{w:spacelike} in~\S\ref{s:kdsu-assumptions}.

Let $\widetilde{\mathcal V}_\pm$ be the one-dimensional subbundles of $T\widetilde\Gamma_\pm$
defined in~\eqref{e:tilde-v-pm}, invariant under the flow $\tilde\varphi^s$.
Since $d\tilde\varphi_\mp$ vanishes on $T\widetilde K$ and is not identically zero
on $T_{\widetilde K}\widetilde\Gamma_\pm$, we can fix a basis $v_\pm$ of
$\widetilde{\mathcal V}_\pm|_{\widetilde K}$ by requiring that
$$
d\tilde\varphi_\mp\cdot v_\pm=1.
$$
Denote by $V=H_G/\partial_\tau G$ the generator of the flow $\tilde\varphi^s$.
The Lie derivative
$\mathcal L_V v_\pm$ is a multiple of $v_\pm$; to compute it, we use the identity
$$
0=V(d\tilde\varphi_\mp\cdot v_\pm)=\mathcal L_V(d\tilde\varphi_\mp)\cdot v_\pm
+d\tilde\varphi_\mp\cdot \mathcal L_V v_\pm.
$$
Since~\eqref{e:reno} holds on $\widetilde\Gamma_+\cup\widetilde\Gamma_-$, we get
on vectors tangent to $\widetilde\Gamma_\pm$,
$$
\mathcal L_V(d\tilde\varphi_\mp)=d(\pm \tilde\nu_\mp\tilde\varphi_\mp)
=\pm\tilde\nu d\tilde\varphi_\mp\quad\text{on }\widetilde K.
$$
It follows that
$$
\partial_s (d\tilde\varphi^s \,v_\pm)=\pm(\tilde\nu\circ\tilde\varphi^s)v_\pm,
$$
which implies immediately
\begin{prop}
  \label{l:radial-expansion}
The expansion rates defined in~\eqref{e:tilde-nu-min} and~\eqref{e:tilde-nu-max}
are given by
$$
\nu_{\min}=\liminf_{T\to\infty}\inf_{(x,\xi)\in K}\langle \nu\rangle_T,\quad
\nu_{\max}=\limsup_{T\to\infty}\sup_{(x,\xi)\in K}\langle \nu\rangle_T,
$$
where $\tilde\nu>0$ is the function on $\widetilde K$ defined in~\eqref{e:tilde-nu} and
$$
\langle\tilde\nu\rangle_T:={1\over T}\int_0^T \tilde\nu\circ \varphi^s\,ds.
$$
\end{prop}
Together, Proposition~\ref{l:angular-flow} and~\ref{l:radial-expansion} verify
assumption~\eqref{w:nhp} of~\S\ref{s:kdsu-assumptions} and finish the proof
of Proposition~\ref{l:kds-trapping}.

\subsection{Trapping in special cases}
  \label{s:kdsu-trapping-special}

We now establish some properties of the trapped set $\widetilde K$ and the local
expansion rate $\tilde\nu$, defined in~\eqref{e:tilde-nu},
in two special cases. We start with the Schwarzschild(--de Sitter) case~\eqref{e:sds-bound},
when everything can be described explicitly:
\begin{prop}
  \label{l:special-sds}
For $a=0$, we have
\begin{gather}
  \label{e:k-1-sds}
\widetilde K=\Big\{\xi_r=0,\ r=3M,\ \tau\neq 0,\ G_\theta={27M^2\over 1-9\Lambda M^2}\tau^2\Big\},\\
  \label{e:nu-sds}
\tilde\nu={\sqrt{1-9\Lambda M^2}\over 3\sqrt 3 M}.
\end{gather}
\end{prop}
\begin{proof}
We recall from~\eqref{e:psi=0} that
$\widetilde K$ is given by the equations $G=0,\xi_r=0,\Psi=0$, where
$\Psi$ is computed using~\eqref{e:PsiPsi}:
$$
\Psi(r)=2\tau r^2(r-3M).
$$
Since $\tau\neq 0$ on $\widetilde K$ by~\eqref{e:houston},
we see that $\Psi=0$ is equivalent to $r=3M$. Now, $\Delta_r(3M)=3M^2(1-9\Lambda M^2)$;
therefore, $G_r=-{27M^2\over (1-9\Lambda M^2)}\tau^2$ for $\xi_r=0$ and $r=3M$ and
we obtain~\eqref{e:k-1-sds}. Next, by~\eqref{e:Psi-r}, we find
$$
\partial_r^2 G_r=-{r^2\tau\over\Delta_r^2}\partial_r\Psi(r)=
-{18\over (1-9\Lambda M^2)^2}\,\tau^2\quad\text{on }\widetilde K.
$$
Finally, we compute
$$
\partial_{\tau}G=-{54M^2\over 1-9\Lambda M^2}\,\tau\quad\text{on }\widetilde K,
$$
and~\eqref{e:nu-sds} follows.
\end{proof}
We next consider the case when $\Lambda=0$ and $a$ approaches the maximal rotation speed $M$ from below,
calculating the expansion rates on two equators to show that the pinching condition~\eqref{e:kds-pinching} is violated:
\begin{prop}
  \label{l:special-kerr}
Fix $M$ and assume that
$$
\Lambda=0,\quad
a=M-\epsilon,\quad
0<\epsilon\ll 1.
$$
Then $\widetilde K_e$, defined in~\eqref{e:tilde-k-e}, is the union of two conical sets
$$
E_\pm=\{r=R_\pm(\epsilon),\ \xi_r=0,\ \xi_\varphi=F_\pm(\epsilon)\tau,\ \theta=\pi/2,\
\xi_\theta=0,\ \tau\neq 0\},
$$
where $R_+(\epsilon),F_+(\epsilon)$ are smooth functions of $\epsilon$,
$R_-(\epsilon),F_-(\epsilon)$ are smooth functions of $\sqrt\epsilon$, and
(see Figure~\ref{f:extremes})
\begin{equation}
  \label{e:eeee}
\begin{aligned}
R_+(\epsilon)=4M+\mathcal O(\epsilon),&\quad
F_+(\epsilon)=7M+\mathcal O(\epsilon);\\
R_-(\epsilon)=M+\sqrt{8\epsilon M/3}+\mathcal O(\epsilon),&\quad
F_-(\epsilon)=-2M-\sqrt{6\epsilon M}+\mathcal O(\epsilon).
\end{aligned}
\end{equation}
Finally, the expansion rates $\tilde\nu$ defined in~\eqref{e:tilde-nu} are given by
(see also Figure~\ref{f:kdsu-nums}(a) in the introduction)
\begin{equation}
  \label{e:eeeenu}
\tilde\nu={3\sqrt{3}\over 28M}+\mathcal O(\epsilon)\quad\text{on }E_+;\quad
\tilde\nu={\sqrt{\epsilon/2M}\over M}+\mathcal O(\epsilon)\quad\text{on }E_-.
\end{equation}
\end{prop}
%
\begin{figure}
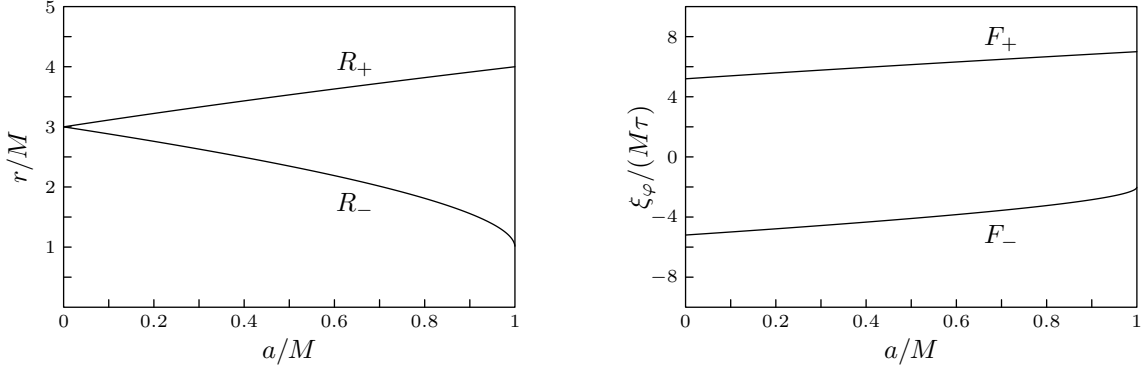

\includegraphics{kdsu.3}
\qquad\quad
\includegraphics{kdsu.4}
\caption{The graphs of $r$ and $\xi_\varphi/\tau$ on the trapped equators $E_\pm$, as functions of $a$
for $\Lambda=0$.}
\label{f:extremes}
\end{figure}
%
\begin{proof}
The set $\widetilde K_e$ is defined by equations $\xi_r=\xi_\theta=0$, $\theta=\pi/2$, and
(see~\eqref{e:psi=0})
\begin{equation}
  \label{e:zazzy}
\begin{gathered}
((r^2+a^2)\tau+a\xi_\varphi)^2=\Delta_r(r)(a\tau+\xi_\varphi)^2,\\
4r\tau\Delta_r(r)=((r^2+a^2)\tau+a\xi_\varphi)\partial_r\Delta_r(r).
\end{gathered}
\end{equation}
Recall that $\Delta_r(r)=r^2+a^2-2Mr$. Putting $A=(r^2+a^2)\tau+a\xi_\varphi$
and $B=a\tau+\xi_\varphi$, we rewrite these as
$$
\begin{gathered}
A^2=\Delta_r(r)B^2,\\
4(A-aB)\Delta_r(r)=Ar\partial_r\Delta_r(r).
\end{gathered}
$$
The second equation can be written as $(r^2+2a^2-3Mr)A=2a\Delta_r(r)B$. Solving
for $B$ and substituting into the first equation, we get
\begin{equation}
  \label{e:lullaby}
4a^2\Delta_r(r)-(r^2+2a^2-3Mr)^2=0.
\end{equation}
This is a fourth order polynomial equation in $r$ with coefficients depending on $\epsilon$
and with a root at $r=0$; we will study the behavior of the other three roots as $\epsilon\to 0$.
We write~\eqref{e:lullaby} as 
\begin{equation}
  \label{e:lullaby2}
(r-M)^2(r-4M)=-8\epsilon M^2+\mathcal O(\epsilon^2).
\end{equation}
By the implicit function theorem, for $\epsilon$ small enough, the equation~\eqref{e:lullaby} has a solution
$R=R_+(\epsilon)=4M+\mathcal O(\epsilon)$. We next identify the two roots lying near $r=M$; they are solutions
to the equations
$$
r-M=\pm M\sqrt{8+\mathcal O(\epsilon)\over 4M-r}\cdot \sqrt{\epsilon}.
$$
The solution with the negative sign lies to the left of $r_->M$, therefore we ignore it.
The solution with the positive sign, which we denote by $R_-(\epsilon)$,
exists for $\epsilon$ small enough by the implicit function theorem and we find
$R_-(\epsilon)=M+\sqrt{8\epsilon M/3}+\mathcal O(\epsilon)$.

To find the values of $\xi_\varphi/\tau$ corresponding to $r=R_\pm(\epsilon)$, we use the second
equation in~\eqref{e:zazzy}; this completes the proof of~\eqref{e:eeee}. Finally, we calculate
at $r=R_+(\epsilon)$, $\xi_\varphi=F_+(\epsilon)\tau$,
$$
\Delta_r=9M^2+\mathcal O(\epsilon),\
\partial_r^2G=-{32\over 3}\tau^2+\mathcal O(\epsilon),\
\partial_{\tau}G=-{224\over 3}M^2\tau+\mathcal O(\epsilon),
$$
and at $r=R_-(\epsilon)$, $\xi_\varphi=F_-(\epsilon)\tau$,
$$
\Delta_r={2M\over 3}\epsilon+\mathcal O(\epsilon^2),\
\partial_r^2G=-{9M\over\epsilon}\tau^2+\mathcal O(\epsilon^{-1/2}),\
\partial_{\tau}G=-2\sqrt{6M\over\epsilon}\,M^2\tau+\mathcal O(1);
$$
\eqref{e:eeeenu} follows.
\end{proof}
We finally explain how to numerically compute the constant $c_{\widetilde K}$ from
the Weyl law of Theorem~\ref{t:kds-weyl}, defined in~\eqref{e:c-k-tilde}.
We can parametrize $\widetilde K\cap \{t=0\}\cap \{\xi_\theta>0\}$ by the variables
$\tau,\xi_\varphi,\theta,\varphi$~-- indeed, we can find $r=r(\tau,\xi_\varphi)$
from the equation $\partial_r G=0$, and find from the equation $G=0$ that
$\xi_\theta=\sqrt{\Theta(\tau,\xi_\varphi,\theta,\varphi)}$, where
$$
\Theta={(1+\alpha)^2\over\Delta_\theta\Delta_r(r(\tau,\xi_\varphi))}((r(\tau,\xi_\varphi)^2+a^2)\tau+a\xi_\varphi)^2
-{(1+\alpha)^2\over\Delta_\theta^2\sin^2\theta}(a\sin^2\theta\,\tau+\xi_\varphi)^2.
$$
The domain of integration is $\{\Theta>0\}\cap \{0\leq\tau\leq 1\}$. Then (using the the symmetry $\xi_\theta\mapsto -\xi_\theta$
and fact that $\Theta$ does not depend on $\varphi$)
$$
c_{\widetilde K}=4\pi\int_{\{\Theta>0\}\cap \{0\leq\tau\leq 1\}} d\theta \wedge d\xi_\theta\wedge d\xi_\varphi
=2\pi\int_{\{\Theta>0\}\cap \{0\leq\tau\leq 1\}}{\partial_\tau\Theta\over\sqrt{\Theta}} \,d\theta d\xi_\varphi d\tau.
$$
Now, $\Theta$ is homogeneous of degree 2 in the $(\tau,\xi_\varphi)$ variables; therefore,
the integrand is homogeneous of degree 0 and we make the change of variables
$\xi_\varphi=s\tau$ to get
\begin{equation}
  \label{e:c-tilde-k-2}
c_{\widetilde K}=\pi \int_{\{\Theta>0\}\cap \{\tau=1\}}{\partial_\tau \Theta\over\sqrt{\Theta}}
\,d\theta d\xi_\varphi.
\end{equation}
We also note that we can compute $\partial_\tau\Theta$ without involving $\partial_\tau r$,
since $\partial_rG=0$ on the trapped set.

For $a=0$, we put $c_0={3\sqrt{3}M\over\sqrt{1-9\Lambda M^2}}$ and compute
(putting $\xi_\varphi=sc_0\sin\theta$)
$$
c_{\widetilde K}=2\pi c_0^2\int_0^\pi\int_{-1}^1 {\sin\theta\over\sqrt{1-s^2}}\,ds d\theta
=4\pi^2 c_0^2.
$$

\subsection{Results for linear waves}
  \label{s:kdsu-waves}

In this section, we apply Theorem~\ref{t:internal-waves} in~\S\ref{s:kdsu-decay}
and the analysis of~\S\S\ref{s:kdsu-the-metric}, \ref{s:kdsu-trapping} to
obtain Theorems~\ref{t:decay} and~\ref{t:res-dec}.

We start by formulating a well-posed problem for the wave equation on the Kerr--de Sitter
background. For that, we in particular need to shift the time
variable, see~\switch{\cite[\S1]{skds} and~\cite[\S1.1]{zeeman}}{\S\S\ref{s:skds-1}, \ref{s:k-ds}}.
Let $\mu$ be the defining function of the event horizons and/or spatial
infinity defined in~\eqref{e:mu-defined} and fix a small constant $\delta_1$,
used in Theorem~\ref{t:internal-waves} as well as in~\eqref{e:fUnz}.
To continue the metric smoothly past the event horizons, we make the change
of variables
\begin{equation}
  \label{e:kds-change}
t=t^*+F_t(r),\quad
\varphi=\varphi^*+F_\varphi(r),
\end{equation}
where $F_t,F_\varphi$ are smooth real-valued functions on $(r_-,r_+)$ such that
\begin{itemize}
\item $F'_t(r)=\pm {1+\alpha\over \Delta_r(r)}(r^2+a^2)+f_\pm(r)$ and
$F'_\varphi(r)=\pm {1+\alpha\over \Delta_r(r)}a$ near $r=r_\pm$, where
$f_\pm$ are smooth functions
(for the Kerr case $\Lambda=0$, we only require this at $r=r_-$)
\item $F_t(r)=F_\varphi(r)=0$ near $\{\mu\geq\delta_1/10\}$
(and also for $r$ large enough in the Kerr case $\Lambda=0$);
\item the covector $dt^*$ is timelike everywhere; equivalently, the level
surfaces of $t^*$ are spacelike.
\end{itemize}
See for example~\cite[\S6.1 and~(6.15)]{vasy1} for how to construct such $F_t,F_\varphi$.
The metric in the coordinates $(t^*,r,\theta,\varphi^*)$ continues smoothly through
$\{r=r_-\}$ and $\{r=r_+\}$ (the latter for $\Lambda>0$), to an extension $\widetilde X_{-{\delta_1}}:=\{\mu>-{\delta_1}\}$
of $\widetilde X_0$ past the event horizons. Since $F_t=F_\varphi=0$ near $\{\mu\geq\delta_1/10\}$,
our change of variables does not affect the arguments in~\S\ref{s:kdsu-general}.

The principal symbol of $h^2\Box_{\tilde g}$ in the new variables, denoted by
$\tilde p^*$, is given by
$$
\tilde p^*(r,\theta,\tau^*,\xi_r,\xi_\theta,\xi_{\varphi^*})
=\tilde p(r,\theta,\tau^*,\xi_r-\partial_r F_t(r)\tau^*-\partial_r F_\varphi(r)\xi_{\varphi^*},
\xi_\theta,\xi_{\varphi^*}).
$$
In particular, if $\xi_{\theta}=\xi_{\varphi^*}=0$, then for
$r$ close to $r_-$ or to $r_+$ (the latter case for $\Lambda>0$),
$$
\tilde p^*=\Delta_r(\xi_r-f_\pm(r)\tau^*)^2
\mp 2(1+\alpha)(r^2+a^2)\tau^*(\xi_r-f_\pm(r)\tau^*)
+{(1+\alpha)^2a^2\sin^2\theta\over\Delta_\theta}(\tau^*)^2.
$$
Then in the new coordinates,
\begin{equation}
  \label{e:naomi}
\tilde g^{-1}(dr,dr)=-\Delta_r,\quad
\tilde g^{-1}(dr,dt^*)=\pm (1+\alpha)(r^2+a^2)+\Delta_r f_\pm(r).
\end{equation}
Therefore, the surfaces $\{r=r_0\}$ are timelike for $\mu(r_0)>0$,
lightlike for $\mu(r_0)=0$, and spacelike for $\mu(r_0)<0$,
and $\tilde g^{-1}(d\mu,dt^*)<0$ near the event horizon(s).
Moreover, for $\Lambda=0$ the d'Alembert--Beltrami operator
$$
\begin{gathered}
\Box_{\tilde g}={1\over \rho^2}D_r(\Delta_r D_r)+{1\over \rho^2\sin\theta}D_\theta(\sin\theta D_\theta)\\
+{(a\sin^2\theta\,D_t+D_\varphi)^2\over\rho^2\sin^2\theta}-{((r^2+a^2)D_t+aD_\varphi)^2\over\rho^2\Delta_r}
\end{gathered}
$$
belongs to Melrose's scattering calculus on the space slices near $r=\infty$ (see~\cite[\S2]{va-zw})
in the sense that it is a polynomial in the differential operators
$D_t,D_r,r^{-1}D_\theta,r^{-1}D_\varphi$
with coefficients smooth
up to $\{r^{-1}=0\}$ in the $r^{-1},\theta,\varphi$ variables (where of course $\theta,\varphi$ are replaced
by a different coordinate system on $\mathbb S^2$ near the poles $\{\sin\theta=0\}$).

Consider the initial-value problem for the wave equation (here $s\geq 0$
is integer)
\begin{equation}
\label{e:wewewe}
\begin{gathered}
\Box_{\tilde g}u=0,\ t^*\geq 0;\quad
u|_{t^*=0}=f_0,\
\partial_{t^*}u|_{t^*=0}=f_1;\\
f_0\in H^{s+1}(X_{-\delta_1}),\
f_1\in H^s(X_{-\delta_1}).
\end{gathered}
\end{equation}
This problem is well-posed, based on standard methods
for hyperbolic equations~\cite[\S6.5]{tay1} and the following crude energy estimate:
if we consider functions on $\widetilde X_{-\delta_1}$ as functions
of $t^*$ with values in functions on $X_{-\delta_1}$, then
for $t'\geq 0$,
\begin{equation}
  \label{e:crude-energy-bound}
\begin{gathered}
\|u(t')\|_{H^{s+1}(X_{-\delta_1})}
+\|\partial_{t^*}u(t')\|_{H^s(X_{-\delta_1})}
\\\leq Ce^{Ct'}(\|u(0)\|_{H^{s+1}(X_{-\delta_1})}
+\|\partial_{t^*}u(0)\|_{H^s(X_{-\delta_1})}
+\|e^{-Ct^*}\Box_{\tilde g}u\|_{H^s((0,t')\times X_{-\delta_1})}).
\end{gathered}
\end{equation}
To prove~\eqref{e:crude-energy-bound}
for $s=0$, we use the standard energy estimate on $\Omega=\widetilde X_{-\delta_1}\cap \{0\leq t^*\leq t'\}$ for hyperbolic equations
(see~\cite[\S2.8]{tay1}, \switch{\cite[Proposition~1.1]{skds}}
{Proposition~\ref{l:crude-energy-estimate}}, or~\cite[\S1.1]{xpd}), with the timelike vector field
$\mathcal N$ equal to $\partial_t$ (a Killing field)
for large $r$ (in the case $\Lambda=0$)
and to $\tilde g^{-1}(dt^*)$ close to the event horizon(s);
by~\eqref{e:naomi}, the boundary $\partial\Omega$ is spacelike
and $\mathcal N$ points inside of $\Omega$ on $\{t^*=0\}$ and 
outside of it elsewhere on $\partial\Omega$. The higher
order estimates are obtained from here as in~\cite[(6.5.14)]{tay1},
commuting with differential operators in the scattering calculus.

We now assume that $f_0=f_0(h),f_1=f_1(h)$ are such
that $\|f_0\|_{H^1(X_{-\delta_1})}+\|f_1\|_{L^2(X_{-\delta_1})}$ is
bounded polynomially in $h$ and $f_0,f_1$ are localized at frequencies $\sim h^{-1}$,
namely (see the discussion in~\S\ref{s:kdsu-prelims})
$$
\WFh(f_0)\cup\WFh(f_1)\subset T^* X_{-\delta_1}\setminus 0.
$$
Let $u$ be the corresponding solution to~\eqref{e:wewewe} and assume that
it is extended to small negative times (which can be done by taking a smaller $\delta_1$
and using the local existence result backwards in time). By~\eqref{e:crude-energy-bound},
we see that $u$ is $h$-tempered uniformly for $t^*\in [0,T\log(1/h)]$. Similarly to~\eqref{e:woof},
$\widetilde\WFh(u)\subset \{\tilde p^*=0\}$. Moreover, using
standard microlocal analysis for hyperbolic equations,
we get a pseudodifferential one-to-one correspondence between $(f_0,f_1)$
and $(u_+(0),u_-(0))$, where $u_\pm$ are the components
of $u$ microlocalized on $\mathcal C_\pm$, the positive and negative parts
of the light cone, each solving an equation of
the form $(hD_t+P_\pm) u_\pm=\mathcal O(h^\infty)$ for some
spatial pseudodifferential operators $P_\pm$ (similarly to~\eqref{e:fac-2}).
This gives
$$
\WFh(u)\cap \{t^*=0\}\subset \{(0,x,\tau,\xi)\mid \tilde p^*(x,\tau,\xi)=0,\
(x,\xi)\in \WFh(f_0)\cup\WFh(f_1)\}.
$$
In particular, we get
\begin{equation}
  \label{e:morrison}
\WFh(u)\cap \{t^*=0\}\subset T^* \widetilde X_{-\delta_1}\setminus 0.
\end{equation}
By the same correspondence, if $\WFh(u)\cap \{t^*=0\}$ is compact and
covered by finitely many open
subsets of $T^*\widetilde X_{-\delta_1}\setminus 0$, then
we can apply the associated pseudodifferential partition of unity to $f_0,f_1$
to split $u$ into several
solutions to the wave equation such that the wavefront set of each solution
at $t^*=0$ is contained in one of the covering sets. The resulting solutions can then each
be analysed separately.

We next assume that
$$
\supp f_0\cup\supp f_1\subset X_{\delta_1}.
$$
We obtain some restrictions on the microlocalization of $u$ for large times.
For that, we need to consider the dynamics of the geodesic flow on the extended spacetime
$\widetilde X_{\delta_1}$. Define the flow~$\hat\varphi^{s}$ similarly
to~\eqref{e:rescaled-flow}, rescaling the geodesic flow so that
the variable $t^*$ is growing with speed 1. Since $t=t^*,\varphi=\varphi^*$ on
$\widetilde X_{\delta_1/10}$, the flow lines of $\tilde\varphi^s$ and
$\hat\varphi^s$ coincide on $\widetilde X_{\delta_1/10}$. If
$\gamma(s)$ is a flow line of $\hat\varphi^s$ such that
$\gamma(0)\in\widetilde X_{\delta_1/10}$ and $\gamma$ is not trapped
for positive times according to Definition~\ref{d:time-trapping},
then either $\gamma(s)$ escapes to the Euclidean infinity (for $\Lambda=0$)
or $\gamma(s)$ crosses one of the event horizons at some fixed positive time $s_0$,
and $\mu(\gamma(s))<0$ is strictly decreasing for $s>s_0$
(see the discussion following~\cite[(6.22)]{vasy1}, verifying~\cite[(2.8)]{vasy1});
in the latter case, we say that $\gamma$ escapes through the event horizons.

The next statement makes nontrivial use of the structure of the infinite ends
(in particular, using~\cite{me-scattering,va-zw,da-smoothing} for the asymptotically
Euclidean end for $\Lambda=0$) and is the key step for obtaining
control on the escaping parts of the solution for long times:
\begin{prop}
  \label{l:eh-escape}
Assume that all flow lines of $\hat\varphi^s$ starting on $\WFh(u)\cap \{t^*=0\}$
escape, either to the spatial infinity or through the event horizons. Then there exists $T_0>0$ such that
uniformly in $t^*$,
\begin{equation}
  \label{e:eh-escape}
\|u(t^*)\|_{H^1(X_{\delta_1})}+
\|\partial_{t^*}u(t^*)\|_{L^2(X_{\delta_1})}=\mathcal O(h^\infty),\quad
t^*\in [T_0,T\log(1/h)].
\end{equation}
\end{prop}
\begin{proof}
We first consider the case when $\WFh(u)\cap \{t^*=0\}$ is contained in a small neighborhood
of some $(\tilde x,\tilde\xi)\in T^* X_{\delta_1}\setminus 0$, and,
for $\gamma(s)=\hat\varphi^s(\tilde x,\tilde\xi)$, there exists $T_0>0$
such that $\gamma([0,T_0])\subset \widetilde X_{-3\delta_1/4}$
and $\gamma(T_0)\in \{\mu<-\delta_1/2\}$. By propagation of singularities for the wave equation
(see for example\switch{~\cite[Proposition~3.4]{nhp}}{ Proposition~\ref{l:microhyperbolic}}),
we see that $\WFh(u)\cap \{t^*=T_0\}\subset \{\mu<-\delta_1/2\}$; it follows
that
$$
\|u(T_0)\|_{H^1(X_{-\delta_1/2})}+
\|\partial_{t^*}u(T_0)\|_{L^2(X_{-\delta_1/2})}=\mathcal O(h^\infty).
$$
Then the same bound holds for $t^*\geq T_0$ in place of $T_0$
by~\eqref{e:crude-energy-bound} (replacing $\delta_1$ by $\delta_1/2$).

For the remainder of this proof, we consider the opposite case, when $\Lambda=0$ and each flow line of $\hat\varphi^s$ starting
on $\WFh(u)\cap \{t^*=0\}$ escapes to the spatial infinity. Fix a large
constant $R_1$; we require in particular that $X_{\delta_1}\subset \{r<R_1\}$.
By propagation of singularities,
similarly to the previous paragraph, we may shift the time parameter and assume that 
$\WFh(u)\cap \{t^*=0\}$
is contained in a small neighborhood of some $(\tilde x_0,\tilde\xi_0)\in T^*\widetilde X_0\setminus 0$,
where $r_0>R_1$, $\partial_s\mu(\hat\varphi^s(\tilde x_0,\tilde\xi_0))|_{s=0}<0$.
In fact, by~\eqref{e:crude-energy-bound} and finite speed of propagation, we may assume that
for $t^*$ near 0, the support of $u$ in $x$ is contained in a compact subset of $\{r>R_1\}$.
Without loss of generality, we assume that $\tau_0<0$.
The trajectory $\hat\varphi^s(\tilde x_0,\tilde\xi_0)$ does not intersect $\{r\leq R_1\}$
for $s\geq 0$.

We replace the Kerr spacetime $(\widetilde X_0,\tilde g)$
with a different spacetime $(\mathbb R_t\times\mathbb R^3_{r,\theta,\varphi},\tilde g_1)$,
where $(r,\theta,\varphi)$ are the spherical coordinates on $\mathbb R^3$ and $\tilde g_1$ is the stationary Lorentzian
metric defined on $\mathbb R^4$ by
$$
\tilde g_1^{-1}:=\chi_1(r)\tilde g_0^{-1}+(1-\chi_1(r))\tilde g^{-1},
$$
where $\tilde g_0^{-1}=\tau^2-\xi_r^2-\xi_\theta^2/r^2-\xi_\varphi^2/(r^2\sin^2\theta)$ is the Minkowski metric on $\mathbb R^4$,
$\chi_1\in C_0^\infty([0,R_1))$, $0\leq\chi_1\leq 1$ everywhere, and $\chi_1=1$ on $[0,R_1/2]$.
The dual metrics $\tilde g^{-1}$ and $\tilde g_0^{-1}$ are close to each other for large $r$
in the sense of scattering metrics, that is, as quadratic forms in $\tau,\xi_r,r^{-1}\xi_\theta,r^{-1}\xi_\varphi$,
therefore for $R_1$ large enough, $\tilde g_1^{-1}$ is the dual to a Lorentzian metric,
the surfaces $\{t=\const\}$ are spacelike, and $\partial_t$ is a timelike vector field.
Note that the new spacetime no longer contains an event horizon.

We now show that $\tilde g_1^{-1}$ is nontrapping for large $R_1$ and a correct choice of $\chi_1$,
that is, each lightlike geodesic escapes to the spatial infinity in both time directions.
It suffices to prove that if $\tilde p_1(\tilde x,\tilde\xi)=-\tilde g^{-1}_{1,\tilde x}(\tilde \xi,\tilde\xi)$, then
(compare with assumption~\eqref{w:convex} in~\S\ref{s:kdsu-assumptions})
$$
r>0,\ \tilde p_1=0,\ \tilde\xi\neq 0,\ H_{\tilde p_1}r(\tilde x,\tilde\xi)=0\ \Longrightarrow\
H_{\tilde p_1}^2 r(\tilde x,\tilde\xi)>0.
$$
Indeed,
$$
H_{\tilde p_1}r=2\xi_r(\chi_1(r)+(1-\chi_1(r))\Delta_r/\rho^2);
$$
therefore, $H_{\tilde p_1}r=0$ implies $\xi_r=0$ and $H_{\tilde p_1}^2 r$ has the same sign as
$$
-\partial_r \tilde p_1=-\chi'_1(r)(\tilde p_0-\tilde p)-\chi_1(r)\partial_r \tilde p_0-(1-\chi_1(r))\partial_r\tilde p;
$$
it remains to note that we can take $r\chi'_1(r)$ bounded by 3, $\tilde p_0-\tilde p$ is small for large
$r$ in the sense of scattering metrics,
and both $r\partial_r \tilde p_0$ and $r\partial_r\tilde p$
are homogeneous of degree 2 in $\tilde\xi$ and bounded from above by a negative
constant on $\{\tau^2+r^{-2}\xi_\theta^2+r^{-2}\xi_\varphi^2=1\}\cap \{\tilde p_1=\xi_r=0\}$,
uniformly in $r^{-1}\geq 0$ for $\partial_r\tilde p_0$ and uniformly in $r^{-1}\in [0,\delta_1)$
for $\partial_r\tilde p$.

Let $u_1$ be the solution to the wave equation on the new spacetime $(\mathbb R^4,\tilde g_1)$
such that $u_1|_{t=0}=u|_{t^*=0}$, $\partial_t u_1|_{t=0}=\partial_{t^*}u|_{t^*=0}$.
It is enough to prove that, with $\widetilde\WFh(u_1)$ defined in~\S\ref{s:kdsu-prelims},
\begin{equation}
  \label{e:ultimate-wf-statement}
\widetilde\WFh(u_1)\cap \{r<R_1\}=\emptyset,\quad
0\leq t\leq T\log(1/h).
\end{equation}
Indeed, in this case $\Box_{\tilde g}((1-\chi_1(r))u_1)=\mathcal O(h^\infty)_{C^\infty}$;
by~\eqref{e:crude-energy-bound}, we have $\widetilde\WFh(u_1)=\widetilde\WFh(u)$
for $t\in [0,T\log(1/h)]$, and~\eqref{e:eh-escape} follows since $X_{\delta_1}\subset \{r<R_1\}$.

To show~\eqref{e:ultimate-wf-statement}, we use the Fourier transform in time,
$$
\hat u_1(\lambda)=\int_{\mathbb R}e^{i\lambda t}\psi_1(t) u_1(t)\,dt,\quad
\Im\lambda>0.
$$
Here $\psi_1(t)$ is supported in $[-\delta,\infty)$ and is equal to $1$ on $[\delta,\infty)$,
for some small fixed $\delta>0$.
The integral converges, as $\|u_1(t)\|_{H^2(\mathbb R^3)}\leq Ce^{\varepsilon t}$
for each $\varepsilon>0$, as follows from the standard energy estimate for the
wave equation (see the paragraph following~\eqref{e:crude-energy-bound}) applied
for the timelike Killing vector field $\partial_t$. 

Let $\widehat P(\lambda)$ be the stationary d'Alembert--Beltrami operator
for the metric $\hat g$, constructed by replacing $D_t$ by $-\lambda$ in the operator $\Box_{\hat g}$;
the semiclassical version defined in~\S\ref{s:kdsu-space} is given by the relation
$\widehat P_h(\omega)=h^2\widehat P(h^{-1}\omega)$. Then
$$
\widehat P(\lambda)\hat u_1(\lambda)=\hat f_1(\lambda),\quad
\Im\lambda>0,
$$
where $f_1=[\Box_{\hat g},\psi_1(t)]u_1(t)$. We note that
$\WFh(f_1)$ is contained in a small neighborhood of $(\tilde x_0,\tilde \xi_0)$
and $f_1$ is compactly supported; therefore,
$\hat f_1(h^{-1}\omega+iE)=\mathcal O(h^\infty)_{\mathscr S_\omega C_0^\infty(\mathbb R^3)}$
for $\omega$ outside of a small neighborhood of $-\tau_0>0$, and
$\WFh(\hat f_1(h^{-1}\omega+iE))$ lies in a small neighborhood of $(x_0,\xi_0)$
for all $\omega$.

We now apply the results of~\cite{me-scattering,va-zw,da-smoothing}.
For this, note that for any fixed $\lambda$, the operator
$\widehat P(\lambda)$ lies in Melrose's scattering calculus
on the radially compactified $\mathbb R^3$, and for $\Im\lambda>0$,
the operator $\widehat P(\lambda)$ is elliptic in this calculus in the microlocal sense
(that is, elliptic as $\xi$ and/or $r$ go to infinity)~--
in fact, near $r=\infty$ the operator $\widehat P(\lambda)$ is close to $\Delta_0 -\lambda^2$,
where $\Delta_0$ is the flat Laplacian on $\mathbb R^3$. Moreover, $\widehat P(\lambda)$ is a symmetric operator
when $\lambda\in\mathbb R$. This implies that the proofs of~\cite{va-zw,da-smoothing} apply.
Similarly to~\cite[Theorem~2]{me-scattering}, for $\Im \lambda>0$, the operator $\widehat P(\lambda)$ is Fredholm
$H^2(\mathbb R^3)\to L^2(\mathbb R^3)$ and invertible for $\lambda$ outside of a discrete set; we can then fix $E>0$ such that
$\widehat P(\lambda+iE)$ is invertible for all $\lambda\in\mathbb R$.

Next, the Hamiltonian flow of the principal symbol $\hat{\mathbf p}(\omega)$
of $\widehat P_h(\omega)$ corresponds to lightlike geodesics of the metric
$\hat g$, similarly to~\eqref{e:rescaled-flow-2}. Therefore, this flow is nontrapping
at all energies $\omega\neq 0$. By~\cite{va-zw},
we get for each $\chi_0\in C_0^\infty(\mathbb R^3)$,
\begin{equation}
  \label{e:hat-r-bobound}
\|\chi_0\widehat P(\lambda+iE)^{-1}\chi_0\|_{L^2\to L^2}\leq C\langle\lambda\rangle^{-1},\quad
\lambda\in\mathbb R;
\end{equation}
in fact, the constant in the estimate is bounded as $E\to 0$, but we do not use this here.
Finally, by~\cite[Lemma~2]{da-smoothing}, we see that
$\widehat P_h(\omega+ihE)^{-1}$ is semiclassically outgoing for $\omega$ near $-\tau_0$,
that is, the wavefront set of $\hat u_1(h^{-1}\omega+iE)$
is contained in the union of $\WFh(\hat f_1(h^{-1}\omega+iE))$ and all Hamiltonian flow lines of 
$\hat{\mathbf p}(\omega)$ starting on $\WFh(\hat f_1(h^{-1}\omega+iE))\cap \{\hat{\mathbf p}(\omega)=0\}$.
Since no geodesic starting near $(\tilde x_0,\tilde\xi_0)$ intersects $\{r\leq R_1\}$
for positive times, we get
$\WFh(\hat u_1(h^{-1}\omega+iE))\cap \overline T^*X_{\delta_1}=\emptyset$ for $\omega$ in a neighborhood
of $-\tau_0$. For $\omega$ outside of this neighborhood, we use the rapid decay
of $\hat f_1(\omega)$ established before, together with~\eqref{e:hat-r-bobound},
to get
$$
\hat u_1(\lambda+iE)=\mathcal O(h^\infty\langle\lambda\rangle^{-\infty})_{C^\infty(\{r<R_1\})}.
$$
It remains to use the Fourier inversion formula
$$
u_1(t)={1\over 2\pi}\int_{\mathbb R}e^{-i(\lambda+iE)t}\hat u_1(\lambda+iE)\,d\lambda
$$
to get~\eqref{e:ultimate-wf-statement}.
\end{proof}
Any solution satisfying~\eqref{e:eh-escape} is trivial from the point
of view of Theorems~\ref{t:decay} and~\ref{t:res-dec} (putting $u_\Pi=0$).
Therefore, we may assume that $\WFh(u)\cap \{t^*=0\}$ is contained in a small neighborhood
of some $(\tilde x,\tilde\xi)$ such that the corresponding geodesic does not escape.
By assumption~\eqref{w:positive} in~\S\ref{s:kdsu-assumptions},
see also Lemma~\ref{l:W-construction}, we may assume that
$$
\WFh(u)\cap \{t^*=0\}\subset (\mathcal C_+\cap \{\tau<0\})\cup (\mathcal C_-\cap \{\tau>0\}),
$$
here $\mathcal C_\pm$ are defined in~\eqref{e:time-orientation}. We can reduce
the case $\WF(u)\cap \{t^*=0\}\subset \mathcal C_-$
to the case $\WFh(u)\cap \{t^*=0\}\subset\mathcal C_+$ by taking
the complex conjugate of $u$, and take a dyadic partition of unity
together with the natural rescaling of the problem $\tilde\xi\mapsto s\tilde \xi, h\mapsto sh$,
to reduce to the case
\begin{equation}
  \label{e:time-is-set}
\WFh(u)\cap \{t^*=0\}\subset \mathcal C_+\cap \{|1+\tau|<\delta_1/8\}.
\end{equation}
\begin{prop}
  \label{l:time-persists}
For $\widetilde\WFh(u)$ defined in~\S\ref{s:kdsu-prelims}, we have
\begin{equation}
  \label{e:time-persists}
\widetilde\WFh(u)\subset \{|1+\tau|<\delta_1/4\}.
\end{equation}
\end{prop}
\begin{proof}
Consider a function $\psi\in C_0^\infty(-1-\delta_1/2,-1+\delta_1/2)$
such that $\psi=1$ near $[-1-\delta_1/4,-1+\delta_1/4]$. If $u$ solves the wave
equation on $(-\delta,\infty)$, then we extend it to a function on the whole
$\widetilde X_{-\delta_1}$ smoothly and so that $\supp u\subset \{t^*>-2\delta\}$.
Define
$$
u':=(1-\psi(hD_{t^*}))u,
$$
then, since the metric is stationary,
$\Box_{\tilde g}u'=(1-\psi(hD_{t^*}))\Box_{\tilde g}u=\mathcal O(h^\infty)_{\mathscr S(\widetilde X_{-\delta_1}\cap \{t^*\geq -\delta/2\})}$.
By~\eqref{e:time-is-set}, we get $\WFh(u')\cap \{t^*=0\}=\emptyset$.
Then by the energy estimate~\eqref{e:crude-energy-bound}, applied to $u'$, we get
$\widetilde\WFh(u')=\emptyset$, uniformly in $t^*\in [0,T\log(1/h)]$.
It remains to note that $\widetilde\WFh(\psi(hD_{t^*})u)\subset \{|1+\tau|<\delta_1/4\}$.
\end{proof}
We can now give
\begin{proof}[Proofs of Theorems~\ref{t:decay} and~\ref{t:res-dec}]
Without loss of generality (replacing $\delta_1$ by $\delta_1/3$)
we may assume that $\supp f_0\cup\supp f_1\subset X_{3\delta_1}$.

Choose small $t_\varepsilon>0$ and a cutoff function
$\chi=\chi(\mu)$, with $\supp\chi\subset \{\mu>2\delta_1\}$
and $\chi=1$ near $\{\mu\geq 3\delta_1\}$, such that,
with the flow $\tilde\varphi^t$ defined in~\eqref{e:rescaled-flow},
\begin{equation}
  \label{e:chihc}
(\tilde x,\tilde\xi)\in\supp\chi,\
\tilde\varphi^{t_\varepsilon}(\tilde x,\tilde\xi)\in\supp(1-\chi),\
\tilde\xi\neq 0\
\Longrightarrow\ {H_{\tilde p}\over \partial_\tau \tilde p}\mu(\tilde\varphi^{t_{\varepsilon}}(\tilde x,\tilde\xi))<0.
\end{equation}
The existence of such $\chi$ and $t_\varepsilon$ follows from Proposition~\ref{l:mu-conv},
see the proof of~\cite[Lemma~5.5(1)]{qeefun}.

Take $N(h)=\lceil T\log(1/h)/t_\varepsilon\rceil$ and consider the functions
$u^{(0)}:=u$ and
$$
\begin{gathered}
u^{(j)}\in C^\infty(\widetilde X_{-\delta_1}\cap \{t^*\geq jt_\varepsilon\}),\quad 1\leq j\leq N(h),\\
\Box_{\tilde g}u^{(j)}=0,\
u^{(j)}(jt_\varepsilon)=\chi u^{(j-1)}(jt_\varepsilon),\
\partial_{t^*}u^{(j)}(jt_\varepsilon)=\chi \partial_{t^*}u^{(j-1)}(jt_\varepsilon).
\end{gathered}
$$
By~\eqref{e:crude-energy-bound}, $u^{(j)}$ are $h$-tempered uniformly in $j$ and
in $t^*\in [jt_\varepsilon,T\log(1/h)+2]$. Moreover,
similarly to Proposition~\ref{l:time-persists}, we get
$\widetilde\WFh(u^{(j)})\subset \{|1+\tau|<\delta_1/4\}$ uniformly in $j$.
Then, $u^{(j)}-u^{(j-1)}$ are solutions to the wave equation
with initial data $(1-\chi)(u^{(j-1)}(jt_\varepsilon),\partial_{t^*}u^{(j-1)}(jt_\varepsilon))$,
therefore by~\eqref{e:chihc}
$$
\WFh(u^{(j)}-u^{(j-1)})\cap \{t^*=jt_\varepsilon\}
\subset \{|1+\tau|<\delta_1/4\}\cap \{\mu>\delta_1\}\cap \bigg\{{H_{\tilde p}\over\partial_\tau \tilde p}\mu<0\bigg\}.
$$
Then all the trajectories of $\hat\varphi^s$ starting on $\WFh(u^{(j)}-u^{(j-1)})\cap\{t^*=jt_\varepsilon\}$ escape
as $s\to +\infty$; by Proposition~\ref{l:eh-escape}, we see that
$$
\widetilde\WFh(u^{(j)}-u^{(j-1)})\cap \{\mu>\delta_1\}=\emptyset,\quad
t^*\in [jt_\varepsilon+T_0,T\log(1/h)],
$$
uniformly in $j$, where $T_0$ is a fixed large constant. Adding these up, we get
\begin{equation}
  \label{e:argentina}
\widetilde\WFh(u^{(j)}-u)\cap \{\mu>\delta_1\}=\emptyset,\quad
t^*\in [jt_\varepsilon+T_0,T\log(1/h)].
\end{equation}
By propagation of singularities for the wave equation
and using that $\WFh(u^{(j)})\cap \{t^*=jt_\varepsilon\}\subset \{\mu>2\delta_1\}$,
we see, uniformly in $j$,
$$
\WFh(u)\cap \{\delta_1\leq \mu\leq 2\delta_1\}\cap \{jt_\varepsilon\leq t^*-T_0\leq (j+1)t_\varepsilon\}\subset \{|1+\tau|<\delta_1/4\}\cap \bigg\{{H_{\tilde p}\over\partial_\tau \tilde p}\mu<0\bigg\}.
$$
Combining this with~\eqref{e:argentina} (and another application of propagation of singularities
for times up to $T_0$), we get uniformly in $t^*\in [0,T\log(1/h)]$,
\begin{equation}
  \label{e:chimp}
\widetilde\WFh(u)\cap \{\delta_1\leq \mu\leq 2\delta_1\}\subset \{|1+\tau|<\delta_1/4\}\cap \bigg\{{H_{\tilde p}\over\partial_\tau \tilde p}\mu<0\bigg\}.
\end{equation}
This implies that for any bounded fixed $T_1$, the semiclassical singularities of $u(t+T_1)$
in $X_{\delta_1}$ come via propagation of singularities from the semiclassical singularities
of $u(t)$ in $X_{\delta_1}$~-- that is, no new singularities arrive from the outside.
We can then apply propagation of singularities
to see that $\widetilde\WFh(u)\cap \{\mu>\delta_1\}\subset\mathcal W$
uniformly in $t^*\in [T_0,T\log(1/h)]$,
where $\mathcal W\subset\mathcal C_+$ is constructed in Lemma~\ref{l:W-construction};
indeed, every trajectory of $\tilde\varphi^s$
starting on $\{|1+\tau|<\delta_1/4\}\cap \{\mu>\delta_1\}\setminus\mathcal W$ escapes
as $s\to +\infty$. Together with~\eqref{e:time-persists} and~\eqref{e:chimp}, this implies that
for $t\geq T_0$, $u$ satisfies the outgoing condition of Definition~\ref{d:we-outgoing}.

We can finally apply Theorem~\ref{t:internal-waves} in~\S\ref{s:kdsu-decay},
giving Theorem~\ref{t:res-dec} and additionally the bounds
(the first one of which is a combination of~\eqref{e:ini-1}, \eqref{e:ini-2},
and~\eqref{e:ini-4})
$$
\begin{gathered}
\|u(t)\|_{\mathcal E}\leq C(h^{-1/2}e^{-(\nu_{\min}-\varepsilon)t/2}+
h^{-1}e^{-(\nu_{\min}-\varepsilon)t}+h^N)\|u(0)\|_{\mathcal E},\\
\|u(t)\|_{\mathcal E}\leq Ce^{\varepsilon t}\|u(0)\|_{\mathcal E}.
\end{gathered}
$$
The first of these bounds gives Theorem~\ref{t:decay} for $(\nu_{\min}-\varepsilon)t\geq\log(1/h)$;
the second one gives
$$
\|u(t)\|_{\mathcal E}\leq Ch^{-1/2}e^{-(\nu_{\min}-3\varepsilon)t/2}\|u(0)\|_{\mathcal E},\quad
(\nu_{\min}-\varepsilon)t\leq \log(1/h),
$$
which is the bound of Theorem~\ref{t:decay} with $\varepsilon$ replaced by $3\varepsilon$.
\end{proof}

\subsection{Results for resonances}
  \label{s:kdsu-resonances}

In this section, we use the results of\switch{~\cite{nhp}}{ Chapter~\ref{c:nhp}}
together with the analysis of~\S\S\ref{s:kdsu-the-metric}, \ref{s:kdsu-trapping}
to prove Theorem~\ref{t:kds-weyl}. As in the statement of this theorem,
we consider the Kerr--de Sitter case $\Lambda>0$.

We first use~\cite[\S 6]{vasy1} to define resonances for Kerr--de Sitter and put them
into the framework of~\switch{\cite[\S4]{nhp}}{\S\ref{s:framework}}.
We use the change of variables~\eqref{e:kds-change};
the metric in the coordinates $(t^*,r,\theta,\varphi^*)$
continues smoothly through the event horizons
to $\widetilde X_{-\delta_1}=\{\mu>-\delta_1\}$, see~\cite[\S6.1]{vasy1}.

Following~\cite[\S6.2]{vasy1} (but omitting the $\rho^2$ factor), we consider the 
the stationary d'Alembert--Beltrami operator $P(z)$, obtained by replacing
$D_{t^*}$ by $-z\in\mathbb C$ in $\Box_{\tilde g}$.
It is an operator on the space slice
$X_{-\delta_1}=\{\mu >-\delta_1\}_r\times \mathbb S^2_{\theta,\varphi}$.
We consider the semiclassical version
$$
P_{\tilde g}(\omega):=h^2P(h^{-1}\omega),
$$
where $h\to 0$ is a small parameter; this definition agrees
with the one used in~\S\ref{s:kdsu-space}.

Following~\cite[\S6.5]{vasy1}, we embed
$X_{-\delta_1}$ as an open set into a compact manifold without boundary $X$,
extend $P(z)$ to a second order differential operator
on $X$ depending holomorphically on $z$, and construct a complex absorbing
operator $Q(z)\in\Psi^2(X)$, whose Schwartz kernel is supported inside the square
of the nonphysical region $\{\mu<0\}$. Then~\cite[Theorem~1.1]{vasy1}
for $\Im z\geq -C_1$ and $s$ large enough depending on $C_1$, $P(z)-iQ(z)$
is a holomorphic family of Fredholm operators $\mathcal X^s\to H^{s-1}(X)$, where
$$
\mathcal X^s=\{u\in H^s(X)\mid (P(0)-iQ(0))u\in H^{s-1}(X)\},
$$
and resonances are defined as the poles of its inverse. The semiclassical version is
\begin{equation}
  \label{e:semicr}
\begin{gathered}
\mathcal P(\omega):=P_{\tilde g}(\omega)-h^2Q(h^{-1}\omega):\mathcal X^s_h\to H^{s-1}_h(X),\\
\|u\|_{\mathcal X^s_h}=\|u\|_{H^s_h(X)}+\|(P(0)-iQ(0))u\|_{H^{s-1}_h(X)}.
\end{gathered}
\end{equation}
We now claim that the operator $\mathcal P(\omega)$ satisfies all the
assumptions of~\switch{\cite[\S\S4.1, 5.1]{nhp}}{\S\S\ref{s:framework-assumptions},
\ref{s:dynamics}}. Most of these assumptions have already been verified in~\S\ref{s:kdsu-space},
relying on the assumptions of~\S\ref{s:kdsu-assumptions} which in turn have
been verified in~\S\S\ref{s:kdsu-the-metric}, \ref{s:kdsu-trapping}.
Given the definition of the spaces $\mathcal H_1:=\mathcal X^s_h$, $\mathcal H_2:=H^{s-1}_h(X)$,
and the Fredholm property discussed above, it remains to verify
assumptions~\switch{(10) and~(11) of~\cite[\S4.1]{nhp}}{\eqref{a:parametrix}
and~\eqref{a:outgoing} of~\S\ref{s:framework-assumptions}}, namely
the existence of an outgoing parametrix. This is done by
modifying the proof of~\cite[Theorem~2.15]{vasy1} exactly as
at the end of~\switch{\cite[\S4.4]{nhp}}{\S\ref{s:framework-ah}}.

Theorem~\ref{t:kds-weyl} now follows directly
by \switch{\cite[Theorems~1 and~2]{nhp}}{Theorems~\ref{t:gaps}
and~\ref{t:weyl-law}}; the constant
$c_{\widetilde K}$ is given by~\eqref{e:c-k-tilde}.

\subsection{Stability}
  \label{s:kdsu-stability}

We finally discuss stability of Theorems~\ref{t:decay}--\ref{t:kds-weyl},
under perturbations of the metric. We assume that $(\tilde X_0,\tilde g)$ is a Lorentzian
manifold which is a small
smooth metric perturbation of the exact Kerr(--de Sitter) (as described in~\S\ref{s:kdsu-the-metric}
and with $M,\Lambda,a$ in a small neighborhood of either~\eqref{e:sds-bound}
or~\eqref{e:kerr-bound}) and which is moreover stationary (that is, $\partial_t$ is Killing).
For perturbations of Kerr ($\Lambda=0$) spacetime, we moreover assume that our perturbation
coincides with the exact metric for large $r$ (this assumption can be relaxed; in fact,
all we need is for~\eqref{e:crude-energy-bound}
and the analysis in Proposition~\ref{l:eh-escape} to apply,
so we may take a small perturbation in the class of scattering metrics). We also assume that the perturbation
continues smoothly across the event horizons in the coordinates~\eqref{e:kds-change}.
The initial value problem~\eqref{e:wewewe} is then well-posed,
as $\{\mu=-\delta_1\}$ is still spacelike. The results of~\cite{vasy1}
still hold, as discussed in~\cite[\S2.7]{vasy1}.

It remains to verify that the assumptions of~\S\ref{s:kdsu-assumptions} still
hold for the perturbed metric. Assumptions~\eqref{w:basic}--\eqref{w:spacelike}
are obviously true. Assumption~\eqref{w:convex} holds with the same function $\mu$,
at least for $\mu(\gamma(s))\in (\delta,\delta_0)$, where $\delta_0$ is fixed
and $\delta>0$ is small depending on the size of the perturbation; 
we take a small enough perturbation so that $\delta\ll\delta_1$, where $\delta_1>0$
is the constant used in Theorem~\ref{t:internal-waves} in~\S\ref{s:kdsu-decay}
and in~\eqref{e:fUnz}. Then the trapped set $\widetilde K$ for the perturbed
metric is close to the original trapped set, which implies assumption~\eqref{w:positive}.
Finally, the dynamical assumptions~\eqref{w:cr}--\eqref{w:nhp} still hold by
the results of~\cite{hps} and the semicontinuity of $\nu_{\min},\nu_{\max},\mu_{\max}$,
as discussed in~\switch{\cite[\S5.2]{nhp}}{\S\ref{s:stability}}.

\smallsection{Acknowledgements} I would like to thank Maciej Zworski for
plenty of advice and constant encouragement, Andr\'as Vasy for many helpful discussions
concerning~\cite{vasy1}, and Mihalis Dafermos for pointing me to~\cite{ynzzzc,hod}.
This work was partially supported by the NSF grant DMS-1201417.


\def\arXiv#1{\href{http://arxiv.org/abs/#1}{arXiv:#1}}

\end{document}